\documentclass[12pt]{article}
\usepackage{authblk}
\usepackage{amsmath,mathtools,amsfonts,amsthm}
\usepackage{graphicx}
\usepackage{enumerate}
\usepackage{booktabs}
\usepackage[round]{natbib}
\usepackage{url} 

\newcommand{\blind}{1}

\newcommand{\piprod}{\text{\huge$\pi$}}
\newcommand{\Rom}[1]{\MakeUppercase{\romannumeral #1}}

\newcommand{\pr}{\text{Pr}}
\newcommand{\trans}{\text{T}}

\newcommand{\mbf}[1]{\boldsymbol{#1}}

\allowdisplaybreaks

\newtheorem{theorem}{Theorem}

\newtheorem{remark}{Remark}
\newtheorem{condition}{Condition}

\begin{document}

\bibliographystyle{agsm}

\def\spacingset#1{\renewcommand{\baselinestretch}%
{#1}\small\normalsize} \spacingset{1}


\if1\blind
{
  \title{\bf Maximum Likelihood Estimation for Semiparametric Regression Models with Interval-Censored Multi-State Data}
  \author[1]{Yu Gu}
  \author[1]{Donglin Zeng}
  \author[2]{Gerardo Heiss}
  \author[1]{D. Y. Lin\thanks{CONTACT: D. Y. Lin (lin@bios.unc.edu), Department of Biostatistics, University of North Carolina, Chapel Hill, NC 27599-7420, USA.}}
  \affil[1]{Department of Biostatistics, University of North Carolina at Chapel Hill, Chapel Hill, NC 27599, USA}
  \affil[2]{Department of Epidemiology, University of North Carolina at Chapel Hill, Chapel Hill, NC 27599, USA}
  \date{\vspace{-5ex}}
  \maketitle
} \fi

\if0\blind
{
  \bigskip
  \bigskip
  \bigskip
  \begin{center}
    {\Large\bf Maximum Likelihood Estimation for Semiparametric Regression Models with Interval-Censored Multi-State Data}
\end{center}
  \medskip
} \fi

\bigskip
\begin{abstract}
Interval-censored multi-state data arise in many studies of chronic diseases, 
where the health status of a subject can be characterized by a finite number of disease states 
and the transition between any two states is only known to occur over a broad time interval. 
We formulate the effects of potentially time-dependent covariates on multi-state processes 
through semiparametric proportional intensity models with random effects. 
We adopt nonparametric maximum likelihood estimation (NPMLE) under general interval censoring 
and develop a stable expectation-maximization (EM) algorithm. 
We show that the resulting parameter estimators are consistent and that
the finite-dimensional components are asymptotically normal  
with a covariance matrix that attains the semiparametric efficiency bound and
can be consistently estimated through profile likelihood. 
In addition, we demonstrate through extensive simulation studies that the proposed numerical and inferential procedures 
perform well in realistic settings. 
Finally, we provide an application to a major epidemiologic cohort study. 
\end{abstract}

\noindent%
{\it Keywords:}  EM algorithm; Nonparametric likelihood; Proportional intensity; Random effects; Semiparametric efficiency;
Time-dependent covariates
\vfill

\newpage
\spacingset{1.9} 
\section{Introduction}

In many studies of chronic diseases, the health status of a subject can be characterized using a finite number of disease states, 
and the disease history of that subject can be viewed as a multi-state stochastic process. 
For example, an old person may first develop mild cognitive impairment (MCI) and then progress to dementia \citep{flicker1991mild}; 
a patient with chronic obstructive pulmonary disease may progress through four stages of the disease \citep{pauwels2001global}.
It is important to understand how a subject transitions from one state to another over time and 
to incorporate the disease history into medical decision-making. 
It is also of interest to study the associations between risk factors and disease processes. For economic and logistical reasons, subjects can only be examined periodically, such that the state transitions are only known to occur between two successive examinations. Such data are called interval-censored multi-state data.
The fact that none of the transition times are directly observed makes semiparametric regression analysis of such data extremely challenging, both theoretically and computationally.  

Most of the literature on interval-censored multi-state data adopts parametric models for transitions and imposes the time-homogeneous Markov assumption
\citep{kalbfleisch1985analysis,satten1999estimating,cook1999mixed,cook2002generalized,cook2004conditional}.   
Parametric models are restrictive, and the homogeneity assumption is violated in many applications.
Several authors used piecewise constant approximations of transition intensities to allow for time nonhomogeneity \citep{gentleman1994multi,saint2003analysis,ocan2005non,jackson2011multi}. 
Others specified spline functions for transition intensities and then applied
piecewise constant approximations for the likelihood construction  \citep{machado2018flexible,machado2021penalised}.
However, the choices for the number of spline pieces and the change points are arbitrary, 
and the results may be sensitive to these choices.  
When the Markov assumption fails, random effects can be used to accommodate the dependence of transitions. 
\citet{satten1999estimating} and \citet{cook2004conditional} considered random effects in modeling the transitions. 
Their methods make strong assumptions about the distribution of the random effects 
and are only applicable to progressive processes.

In this article, we provide a new framework based on semiparametric proportional intensity models with random effects to study general interval-censored multi-state data. 
Our formulation allows the baseline intensity functions for the transitions between any two states to be completely arbitrary
and accommodates time-dependent covariates. 
In addition, we introduce random effects and their possible interactions with covariates 
to further capture the dependence among the transitions of the same subject.
We adopt the NPMLE approach and develop a stable EM algorithm that involves maximization over only a small number of parameters 
and performs well even with complex transition patterns. 
We establish the asymptotic properties of the parameter estimators
through novel use of modern empirical process theory. 
We compare the performance of the proposed and existing methods
through extensive simulation studies. 
Finally, we apply the proposed methods to data on MCI and dementia from the Atherosclerosis Risk in Communities (ARIC) study \citep{knopman2016mild,wright2021aric}. 


\section{Theory and Methods}

\subsection{Models, Data, and Likelihood}
We consider a multi-state process with $K$ states in a study of $n$ subjects. 
Let $\mathcal{D}$ denote the set of all state pairs $(j,k)$ such that $j\ne k$ and transition from $j$ to $k$ is feasible.
We assume that it is impossible for a subject to return to a prior state through other states; otherwise, there would be infinite many loops between two states within any time interval, which would cause non-identifiability issues. 
For $i=1,\dots,n$, let $\mbf{X}_i(\cdot)$ denote a $d_1$-vector of potentially time-dependent covariates for the $i$th subject,
and $\mbf{b}_i$ denote the corresponding $d_2$-vector of random effects that is normal with mean zero and covariance matrix $\mbf{\Sigma}(\mbf{\gamma})$ indexed by $d_3$-dimensional parameters $\mbf{\gamma}$. 
For $(j,k)\in\mathcal{D}$, let $N_{ijk}(t)$ denote the number of times that the $i$th subject transitions from state $j$ to state $k$ by time $t$.
Under proportional intensity models, the transition intensities of $N_{ijk}(t)$ conditional on $\mbf{X}_i$ and $\mbf{b}_i$ take the form 
\[
\lambda_{ijk}(t; \mbf{X}_i, \mbf{b}_i) = \lambda_{jk}(t)\exp\bigl\{\mbf{\beta}_{jk}^{\trans}\mbf{X}_i(t)+\mbf{b}_i^{\trans}\mbf{Z}_i(t)\bigr\},
\]
where $\mbf{Z}_i(\cdot)$ consists of 1 and covariates that may be part of $\mbf{X}_i(\cdot)$, 
$\mbf{\beta}_{jk}$ is a vector of unknown regression parameters, 
and $\lambda_{jk}(\cdot)$ is an arbitrary baseline intensity function.

We define the $K\times K$ cumulative transition intensity matrix $\mbf{A}_i(t;\mbf{X}_i,\mbf{Z}_i, \mbf{b}_i)$, whose off-diagonal elements are  
\[
\mbf{A}_i(t; \mbf{X}_i,\mbf{Z}_i, \mbf{b}_i)^{(j,k)} = 
\begin{cases}
\int_0^t \exp\bigl\{\mbf{\beta}_{jk}^{\trans}\mbf{X}_i(s)+\mbf{b}_i^{\trans}\mbf{Z}_i(s)\bigr\}d\Lambda_{jk}(s) & \text{ if }  (j,k)\in\mathcal{D}, \\
0 & \text{ otherwise}, \\
\end{cases}
\]
and whose diagonal elements are 
\[
\mbf{A}_i(t; \mbf{X}_i, \mbf{Z}_i,\mbf{b}_i)^{(j,j)} = -\sum_{k\ne j} \mbf{A}_i(t; \mbf{X}_i, \mbf{Z}_i, \mbf{b}_i)^{(j,k)},
\]
where $\Lambda_{jk}(t) = \int_0^t\lambda_{jk}(s)ds$, and we use superscript $(j,k)$ to denote the $(j,k)$th element of a matrix. 
For any $0\le t_1\le t_2$, let $\mbf{P}_i(t_1,t_2;\mbf{X}_i,\mbf{Z}_i,\mbf{b}_i)$ denote the $K\times K$ transition probability matrix over the time interval $(t_1,t_2]$. 
According to Theorem \Rom{2}.6.7 of \citet{andersen1993statistical}, the relationship between the two matrices $\mbf{P}_i$ and $\mbf{A}_i$ can be characterized via product integration:
\[
\mbf{P}_i(t_1,t_2; \mbf{X}_i,\mbf{Z}_i, \mbf{b}_i) = \piprod_{t_1<t\le t_2}\bigl\{\mbf{I}_K+d\mbf{A}_i(t; \mbf{X}_i,\mbf{Z}_i, \mbf{b}_i)\bigr\},
\]
where $\mbf{I}_K$ is the $K\times K$ identity matrix, and $d\mbf{A}(\cdot)$ is the element-wise differential for a matrix-valued function $\mbf{A}(\cdot)$. Here,   
\[
\piprod_{s\in(t_1,t_2]} \left\{\mbf{I}+d\mbf{A}(s)\right\} = \lim_{\max|s_l-s_{l-1}|\rightarrow0}\prod_{l=1}^L \left\{\mbf{I}+\mbf{A}(s_l)-\mbf{A}(s_{l-1})\right\},
\]
where $t_1=s_0<s_1<\cdots<s_L=t_2$ is a partition of $(t_1,t_2]$, and the matrix product is taken in its natural order from left to right \citep{gill1990survey}.

We consider a very general interval-censoring scheme, where every subject can be examined an arbitrary number of times. For $i=1,\dots,n$, let $n_i$ denote the number of examinations after the baseline examination for the $i$th subject, and let $0=\tau_{i0}<\tau_{i1}<\cdots<\tau_{in_i}$ denote the corresponding examination times. The state occupied at each examination is denoted by $S_{il}$, $l=0,\dots,n_i$. Then the observed data consist of 
$\left\{(\tau_{i0},\tau_{i1},\dots,\tau_{in_i}),\,(S_{i0},S_{i1},\dots,S_{in_i}),\, \mbf{X}_i,\, \mbf{Z}_i \right\}$, $i=1,\dots,n$. 
Write $\mbf{\theta} = {(\mbf{\beta}^{\trans}, \mbf{\gamma}^{\trans})}^{\trans} 
= {(\{\mbf{\beta}_{jk}^{\trans}\}_{(j,k)\in\mathcal{D}}, \mbf{\gamma}^{\trans})}^{\trans}$ and 
$\mbf{\Omega} = \{\Lambda_{jk}\}_{(j,k)\in\mathcal{D}}$. 
Under the conditional Markov assumption, the observed-data likelihood conditional on the initial states is given by
\[
L_n(\mbf{\theta}, \mbf{\Omega}) = \prod_{i=1}^n\int_{\mbf{b}_i}\prod_{l=1}^{n_i}\mbf{P}_i(\tau_{i,l-1},\tau_{il}; \mbf{X}_i,\mbf{Z}_i, \mbf{b}_i)^{(S_{i,l-1}, S_{il})}\phi(\mbf{b}_i;\mbf{\Sigma}(\mbf{\gamma}))d\mbf{b}_i,
\]
where $\phi(\mbf{b};\mbf{\Sigma}) = (2\pi)^{-d_2/2}|\mbf{\Sigma}|^{-1/2}\exp(-\mbf{b}^{\trans}\mbf{\Sigma}^{-1}\mbf{b}/2)$.

\subsection{Nonparametric Maximum Likelihood Estimation}
\label{subsec: npmle}

We adopt the NPMLE approach to estimate the parameters $\mbf{\theta}$ and $\mbf{\Omega}$. Specifically, for each $(j,k)\in\mathcal{D}$, we treat $\Lambda_{jk}(\cdot)$ as a step function with nonnegative jumps at $0<u_1<u_2<\cdots<u_m$, which are the unique values of $\tau_{il}$ $(i=1,\dots,n;\, l=1,\dots,n_i)$. For $(j,k)\in\mathcal{D}$ and $s=1,\dots,m$, let $\lambda_{jks}$ denote the jump size of $\Lambda_{jk}$ at $u_s$. Then the transition probability matrix $\mbf{P}_i(t_1,t_2; \mbf{X}_i,\mbf{Z}_i, \mbf{b}_i)$ is equal to
\[
\widetilde{\mbf{P}}_i(t_1,t_2; \mbf{X}_i,\mbf{Z}_i, \mbf{b}_i) =\prod_{t_1<u_s\le t_2}\bigl\{\mbf{I}_K+\Delta \mbf{A}_i(u_s; \mbf{X}_i,\mbf{Z}_i, \mbf{b}_i)\bigr\},
\]
where the elements of the matrix $\Delta \mbf{A}_i(u_s; \mbf{X}_i, \mbf{Z}_i, \mbf{b}_i)$ are given by 
\[ 
\Delta \mbf{A}_i(u_s; \mbf{X}_i, \mbf{Z}_i, \mbf{b}_i)^{(j,k)} = 
\begin{cases}
\lambda_{jks}\exp(\mbf{\beta}_{jk}^{\trans}\mbf{X}_{is}+\mbf{b}_i^{\trans}\mbf{Z}_{is}) & \text{ if } (j,k)\in\mathcal{D}, \\
-\sum_{k':(j,k')\in\mathcal{D}}\lambda_{jk's}\exp(\mbf{\beta}_{jk'}^{\trans}\mbf{X}_{is}+\mbf{b}_i^{\trans}\mbf{Z}_{is}) & \text{ if } j=k, \\
0 & \text{ otherwise, }
\end{cases}
\]
with $\mbf{X}_{is} = \mbf{X}_i(u_s)$ and $\mbf{Z}_{is} = \mbf{Z}_i(u_s)$. We maximize 
\begin{equation} \label{nonparam_likelihood}
\prod_{i=1}^n\int_{\mbf{b}_i}\prod_{l=1}^{n_i}\widetilde{\mbf{P}}_i(\tau_{i,l-1},\tau_{il}; \mbf{X}_i, \mbf{Z}_i,\mbf{b}_i)^{(S_{i,l-1}, S_{il})}\phi(\mbf{b}_i;\mbf{\Sigma}(\mbf{\gamma}))d\mbf{b}_i.
\end{equation}
Direct maximization of \eqref{nonparam_likelihood} is very difficult, since it involves matrix multiplication 
and there are no analytical expressions for $\lambda_{jks}$'s. 
Thus, we introduce latent Poisson random variables whose observed-data likelihood is equal to \eqref{nonparam_likelihood} 
but can be maximized through an EM algorithm.

For $i=1,\dots,n$, $(j,k)\in\mathcal{D}$, and $s=1,\dots,m$, we introduce independent latent Poisson random variables $W_{ijks}$ with means $\lambda_{ijks}= \lambda_{jks}\exp(\mbf{\beta}_{jk}^{\trans}\mbf{X}_{is}+\mbf{b}_i^{\trans}\mbf{Z}_{is})$. For the $i$th subject, let $(\tau_1,\tau_2]$ be any of the time intervals $(\tau_{i,l-1}, \tau_{il}]$, $l=1,\dots,n_i$. The unique time points within $[\tau_1,\tau_2]$ are labeled as $u_{s_0} = \tau_1<u_{s_1}<u_{s_2}\cdots<u_{s_q}< \tau_2 = u_{s_{q+1}}$. A transition from state $S_1$ at $\tau_1$ to state $S_2$ at $\tau_2$ consists of all possible transition paths of the form $(k_0=S_1, k_1,k_2,\dots,k_q,S_2=k_{q+1})$, where $k_1,\dots,k_q$ are the unknown states occupied at $u_{s_1},\dots,u_{s_q}$. Given a feasible path $(S_1, k_1,\dots,k_q,S_2)$, we define the event $V_i(k_1,\dots,k_q;\tau_1,\tau_2,S_1,S_2)$ as follows: for $l=1,\dots,q+1$, if $k_{l-1}\ne k_l$, then $W_{ik_{l-1}k_ls_l}>0$ and $W_{ik_{l-1}k's_l}=0$ for all $k'\ne k_{l-1}, k_l$; otherwise $W_{ik_{l-1}k's_l} = 0$ for all $k'\ne k_{l-1}$. 
We claim that the transition probability from state $S_1$ at $\tau_1$ to state $S_2$ at $\tau_2$ is equal to the probability of observing the following event: 
\[
Y_i(\tau_1,\tau_2,S_1,S_2) = \bigcup_{(k_1,\dots,k_q)\in\mathcal{A}_q}V_i(k_1,\dots,k_q;\tau_1,\tau_2,S_1,S_2),
\]
where $\mathcal{A}_q$ is the set of all possible combinations of $k_1,\dots,k_q$ that connect $S_1$ to $S_2$. 

To see this, we consider a sequence of time points $t_0=\tau_1<t_1<\cdots<t_r<\tau_2=t_{r+1}$ such that there is at most one transition within each time interval $(t_{l-1}, t_l]$, $l=1,\dots,r+1$, and the transition time is not necessarily $t_{l}$. 
Let $j_0, \dots, j_{r+1}$ denote the state occupied at $t_0,\dots, t_{r+1}$. The transition probability can be written as 
\begin{equation} \label{continuous_equiv}
\begin{aligned}
& \mbf{P}_i(\tau_1,\tau_2; \mbf{X}_i, \mbf{Z}_i, \mbf{b}_i)^{(S_1,S_2)} \\
={} &
\begin{multlined}[t]
\sum_{(j_1,\dots, j_r)\in\mathcal{A}_r}\prod_{l=1}^{r+1}\Biggl[\exp\biggl\{-\sum_{k\ne j_{l-1}}\int_{t_{l-1}}^{t_l}d\mbf{A}_i(t; \mbf{X}_i, \mbf{Z}_i, \mbf{b}_i)^{(j_{l-1},k)}\biggr\}\Biggl]^{I(j_{l-1}=j_l)} \\
\times \Biggl[\biggl\{1-\exp\Bigl\{-\int_{t_{l-1}}^{t_l}d\mbf{A}_i(t; \mbf{X}_i, \mbf{Z}_i, \mbf{b}_i)^{(j_{l-1},j_l)}\Bigr\}\biggr\}\Biggr. \\
\Biggl.\times \exp\biggl\{-\sum_{k\ne j_{l-1}, j_l}\int_{t_{l-1}}^{t_l}d\mbf{A}_i(t; \mbf{X}_i, \mbf{Z}_i, \mbf{b}_i)^{(j_{l-1},k)}\biggr\}\Biggr]^{I(j_{l-1}\ne j_l)}.
\end{multlined}
\end{aligned}
\end{equation}
In the NPMLE approach, 
transitions within $(\tau_1, \tau_2]$ can only occur at $u_{s_1}, u_{s_2}, \dots, u_{s_{q+1}}$, which ensures at most one transition within each interval $(u_{s_{l-1}}, u_{s_l}]$, $l=1,\dots,q+1$. Thus, we can replace $\{(t_0,\dots,t_{r+1}), \, (j_0,\dots, j_{r+1})\}$ in \eqref{continuous_equiv} with $\{(u_{s_0},\dots,u_{s_{q+1}}), \, (k_0,\dots, k_{q+1})\}$ and plug in the discretized $\mbf{A}_i$ to obtain 
\begin{align*}
& \widetilde{\mbf{P}}_i(\tau_1,\tau_2; \mbf{X}_i, \mbf{Z}_i, \mbf{b}_i)^{(S_1,S_2)} \\
={} &
\begin{multlined}[t]
\sum_{(k_1,\dots, k_q)\in\mathcal{A}_q}\prod_{l=1}^{q+1}\Biggl[\exp\biggl\{-\sum_{k'\ne k_{l-1}}\Delta\mbf{A}_i(u_{s_l}; \mbf{X}_i, \mbf{Z}_i, \mbf{b}_i)^{(k_{l-1},k')}\biggr\}\Biggl]^{I(k_{l-1}=k_l)} \\
\times \Biggl[\biggl\{1-\exp\Bigl\{-\Delta\mbf{A}_i(u_{s_l}; \mbf{X}_i, \mbf{Z}_i, \mbf{b}_i)^{(k_{l-1},k_l)}\Bigr\}\biggr\}\Biggr. \\
\Biggl.\times \exp\biggl\{-\sum_{k'\ne k_{l-1}, k_l}\Delta\mbf{A}_i(u_{s_l}; \mbf{X}_i, \mbf{Z}_i, \mbf{b}_i)^{(k_{l-1},k')}\biggr\}\Biggr]^{I(k_{l-1}\ne k_l)}
\end{multlined} \\
={} & 
\begin{multlined}[t]
\sum_{(k_1,\dots, k_q)\in\mathcal{A}_q}\prod_{l=1}^{q+1}\Biggl[\exp\biggl\{-\sum_{k'\ne k_{l-1}}\lambda_{ik_{l-1}k's_l}\biggr\}\Biggl]^{I(k_{l-1}=k_l)} \\
\Biggl.\times \Biggl[\biggl\{1-\exp(-\lambda_{ik_{l-1}k_ls_l})\biggr\}\times \exp\biggl\{-\sum_{k'\ne k_{l-1}, k_l}\lambda_{ik_{l-1}k's_l}\biggr\}\Biggr]^{I(k_{l-1}\ne k_l)}.
\end{multlined} 
\end{align*}
We can verify that all the events $V_i(k_1,\dots,k_q;\tau_1,\tau_2,S_1,S_2)$ are mutually exclusive. Thus,
\begin{align*}
& \pr\{Y_i(\tau_1,\tau_2,S_1,S_2)\} \\
={} & \sum_{(k_1,\dots,k_q)\in\mathcal{A}_q}\pr\{V_i(k_1,\dots,k_q;\tau_1,\tau_2,S_1,S_2)\} \\
={} & 
\begin{multlined}[t]
\sum_{(k_1,\dots, k_q)\in\mathcal{A}_q}\prod_{l=1}^{q+1}\Biggl[\exp\biggl\{-\sum_{k'\ne k_{l-1}}\lambda_{ik_{l-1}k's_l}\biggr\}\Biggl]^{I(k_{l-1}=k_l)} \\
\Biggl.\times \Biggl[\biggl\{1-\exp(-\lambda_{ik_{l-1}k_ls_l})\biggr\}\times \exp\biggl\{-\sum_{k'\ne k_{l-1}, k_l}\lambda_{ik_{l-1}k's_l}\biggr\}\Biggr]^{I(k_{l-1}\ne k_l)},
\end{multlined} 
\end{align*}
which equals $\widetilde{\mbf{P}}_i(\tau_1,\tau_2; \mbf{X}_i, \mbf{Z}_i, \mbf{b}_i)^{(S_1,S_2)}$.
Hence, maximizing \eqref{nonparam_likelihood} is tantamount to maximizing the likelihood based on the observations $\mathcal{O}_i = \bigcap_{l=1}^{n_i}Y_i(\tau_{i,l-1},\tau_{il},S_{i,l-1},S_{il})$, $i=1,\dots,n$.

To maximize the latter likelihood, we develop an EM algorithm by 
treating $W_{ijks}$ ($i=1,\dots,n$; $(j,k)\in\mathcal{D}$; $s=1,\dots,m$) and $\mbf{b}_i$ ($i=1,\dots,n$) as missing data. 
The complete-data log-likelihood is 
\begin{equation} \label{complete_likelihood}
\begin{split}
\sum_{i=1}^n\biggl\{\sum_{(j,k)\in\mathcal{D}}\sum_{s=1}^m I(t_s\le \tau_{i,n_i})\Bigl\{W_{ijks}\log\left[\lambda_{jks}\exp(\mbf{\beta}_{jk}^{\trans}\mbf{X}_{is}+\mbf{b}_i^{\trans}\mbf{Z}_{is})\right]\Bigr.\biggr. \\
\biggl.\Bigl.-\lambda_{jks}\exp(\mbf{\beta}_{jk}^{\trans}\mbf{X}_{is}+\mbf{b}_i^{\trans}\mbf{Z}_{is})-\log(W_{ijks}!)\Bigr\} \\
\biggl.-\frac{d_2}{2}\log(2\pi)-\frac{1}{2}\log|\mbf{\Sigma}|-\frac{1}{2}\mbf{b}_i^{\trans}\mbf{\Sigma}^{-1}\mbf{b}_i\biggr\}.
\end{split}
\end{equation}
In the E-step, we calculate the conditional expectations of $W_{ijks}$, 
$\exp(\mbf{\beta}_{jk}^{\trans}\mbf{X}_{is}+\mbf{b}_i^{\trans}\mbf{Z}_{is})$ and $\mbf{b}_i^{\otimes2}$ 
given $\mathcal{O}_i$, where $\mbf{a}^{\otimes2} = \mbf{a}\mbf{a}^{\trans}$ for any vector or matrix $\mbf{a}$. The last two conditional expectations 
can be derived from the fact that the conditional distribution of $\mbf{b}_i$ given $\mathcal{O}_i$ is proportional to 
\[
\prod_{l=1}^{n_i}\widetilde{\mbf{P}}_i(\tau_{i,l-1},\tau_{il}; \mbf{X}_i, \mbf{Z}_i, \mbf{b}_i)^{(S_{i,l-1}, S_{il})}\phi(\mbf{b}_i;\mbf{\Sigma}(\mbf{\gamma})).
\]
In addition, the first conditional expectation follows from the conditional expectation of $W_{ijks}$ given $\mathcal{O}_i$ and $\mbf{b}_i$. Assume that $u_s$ falls within the time interval $(\tau_{i,l-1},\tau_{il}]$. 
As before, we denote all the unique time points within the closed interval $[\tau_{i,l-1},\tau_{il}]$ as $\tau_{i,l-1} = u_{s_0},u_{s_1},\dots,u_{s_q},u_{q+1} = \tau_{il}$, and label their corresponding states as $S_{i,l-1} = k_0, k_1,\dots,k_q, k_{q+1} = S_{il}$, where $k_1,\dots,k_q$ can only take values in the set $\mathcal{A}_q$. 
We omit the argument $y$ in the expression $f(x;y)$ when no ambiguity may arise. 
Then the conditional expectation $E(W_{ijks}\mid \mathcal{O}_i,\mbf{b}_i)$ is equal to 
\begin{align*}
& E\{W_{ijks}\mid Y_i(\tau_{i,l-1},\tau_{il},S_{i,l-1},S_{il}),\mbf{b}_i\} \\
={} & \sum_{w=1}^{\infty}\frac{w\times\pr\left\{(W_{ijks} = w) \cap Y_i(\tau_{i,l-1},\tau_{il},S_{i,l-1},S_{il})\right\}}{\pr\left\{Y_i(\tau_{i,l-1},\tau_{il},S_{i,l-1},S_{il})\right\}} \\
={} & \sum_{w=1}^{\infty}\sum_{(k_1,\dots,k_q)\in\mathcal{A}}w\times\pr\left\{(W_{ijks} = w) \cap V_i(k_1,\dots,k_q)\right\}/\widetilde{\mbf{P}}_i(\tau_{i,l-1},\tau_{il})^{(S_{i,l-1},S_{il})}.
\end{align*}
Suppose that $u_s = u_{s_{l}}$ for some $l$. 
Only those paths with $k_{l-1}\ne j$ or $(k_{l-1},k_l) = (j,k)$ will contribute to the above equation.
Thus, the above conditional expectation becomes
\begin{align*}
& 
\frac{\sum_{j'\ne j}\widetilde{\mbf{P}}_i(\tau_{i,l-1},u_{s-1})^{(S_{i,l-1},j')}\widetilde{\mbf{P}}_i(u_{s-1},\tau_{il})^{(j',S_{il})}}{\widetilde{\mbf{P}}_i(\tau_{i,l-1},\tau_{il})^{(S_{i,l-1},S_{il})}}\times \biggl\{\sum_{w=1}^{\infty}w\times\pr(W_{ijks} = w)\biggr\} \\
+{} & \frac{\widetilde{\mbf{P}}_i(\tau_{i,l-1},u_{s-1})^{(S_{i,l-1},j)}\widetilde{\mbf{P}}_i(u_{s},\tau_{il})^{(k,S_{il})}}{\widetilde{\mbf{P}}_i(\tau_{i,l-1},\tau_{il})^{(S_{i,l-1},S_{il})}}
\times \biggl\{\sum_{w=1}^{\infty}w\times\pr(W_{ijks} = w, W_{ijk's} = 0, k'\ne j,k)\biggr\} \\
={} & 
\begin{multlined}[t]
\frac{\sum_{j'\ne j}\widetilde{\mbf{P}}_i(\tau_{i,l-1},u_{s-1})^{(S_{i,l-1},j')}\widetilde{\mbf{P}}_i(u_{s-1},\tau_{il})^{(j',S_{il})}}{\widetilde{\mbf{P}}_i(\tau_{i,l-1},\tau_{il})^{(S_{i,l-1},S_{il})}}\lambda_{ijks} \\
+\frac{\widetilde{\mbf{P}}_i(\tau_{i,l-1},u_{s-1})^{(S_{i,l-1},j)}\widetilde{\mbf{P}}_i(u_{s},\tau_{il})^{(k,S_{il})}}{\widetilde{\mbf{P}}_i(\tau_{i,l-1},\tau_{il})^{(S_{i,l-1},S_{il})}}\lambda_{ijks}\exp\biggl(-\sum_{k'\ne j,k}\lambda_{ijk's}\biggr).
\end{multlined}
\end{align*}
Finally, we approximate the integrals over $\mbf{b}_i$ using Gaussian-Hermite quadratures.

In the M-step, we update $\lambda_{jks}$ by 
\[
\frac{\sum_{i=1}^n I(u_s\le \tau_{i,n_i})\widetilde{E}(W_{ijks})}{\sum_{i=1}^n I(u_s\le \tau_{i,n_i}) \widetilde{E}\{\exp(\mbf{\beta}_{jk}^{\trans}\mbf{X}_{is}+\mbf{b}_i^{\trans}\mbf{Z}_{is})\}},
\]
for $(j,k)\in\mathcal{D}$ and $s=1,\dots,m$, where $\widetilde{E}(\cdot)$ denotes the conditional expectation given $\mathcal{O}_i$. After plugging in the new $\lambda_{ijks}$ values into \eqref{complete_likelihood}, we solve the following score equation for $\mbf{\beta}_{jk}$ ($(j,k)\in\mathcal{D}$) using the one-step Newton-Raphson method:
\[
\begin{split}
\sum_{i=1}^n\sum_{s=1}^m I(u_s\le \tau_{i,n_i}) \widetilde{E}(W_{ijks})\left[\mbf{X}_{is}-\frac{\sum_{i'=1}^n I(u_s\le \tau_{i',n_{i'}}) \mbf{X}_{i's}\widetilde{E}\{\exp(\mbf{\beta}_{jk}^{\trans}\mbf{X}_{i's}+\mbf{b}_{i'}^{\trans}\mbf{Z}_{i's})\}}{\sum_{i'=1}^n I(u_s\le \tau_{i',n_{i'}}) \widetilde{E}\{\exp(\mbf{\beta}_{jk}^{\trans}\mbf{X}_{i's}+\mbf{b}_{i'}^{\trans}\mbf{Z}_{i's})\}}\right] \\
= \mbf{0}.
\end{split}
\]
Finally, we update $\mbf{\Sigma}$ by $n^{-1}\sum_{i=1}^n \widetilde{E}(\mbf{b}_i^{\otimes 2})$.

We iterate between the E-step and the M-step until convergence. 
Denote the resulting estimators of $\mbf{\theta}$ and $\mbf{\Omega}$
by $\widehat{\mbf{\theta}}=(\widehat{\mbf{\beta}}^{\trans}, \widehat{\mbf{\gamma}}^{\trans})^{\trans}
=(\{\widehat{\mbf{\beta}}_{jk}^{\trans}\}_{(j,k)\in\mathcal{D}},\widehat{\mbf{\gamma}}^{\trans})^{\trans}$ 
and $\widehat{\mbf{\Omega}}=\{\widehat{\Lambda}_{jk}\}_{(j,k)\in\mathcal{D}}$.
  
In the M-step, the jump sizes $\lambda_{jks}$'s are updated through explicit expressions, 
so optimization over a large number of parameters is avoided.  
When $n$ and $m$ are very large, the EM algorithm can be demanding, 
since it needs to perform matrix multiplication many times. Thus, we provide some strategies to speed up the computation. 
First, we estimate the jump sizes using \citet{turnbull1976empirical}'s method
and remove those time points with estimates smaller than a threshold of the order $1/m$. 
Second, we set the initial parameter estimates to be the convergent values from the EM algorithm without random effects.
Finally, we remove the time points where the jump sizes are smaller than a threshold of the order $1/m$.
Our experiences showed that all these strategies can significantly reduce 
the computation time without impairing estimation accuracy.

\section{Asymptotic Theory}

Let $|\mathcal{D}|$ denote the cardinality of $\mathcal{D}$. 
We establish the asymptotic properties of $(\widehat{\mbf{\theta}},\widehat{\mbf{\Omega}})$ under the following regularity conditions. We consider a generic subject and omit the subscript $i$ in all random quantities.   

\begin{condition}\label{cond1}
The true value of $\mbf{\theta}$, denoted by $\mbf{\theta}_0=(\mbf{\beta}_0^{\trans},\mbf{\gamma}_0^{\trans})^{\trans}
=(\{\mbf{\beta}_{0jk}^{\trans}\}_{(j,k)\in\mathcal{D}},\mbf{\gamma}_0^{\trans})^{\trans}$, 
lies in the interior of a known compact set $\Theta=\{(\mbf{\beta}^{\trans},\mbf{\gamma}^{\trans})^{\trans}:\; 
\mbf{\beta}\in\mathcal{B},\, \mbf{\gamma}\in\mathcal{C}\}$, 
where $\mathcal{B}$ is a compact set in $\mathbb{R}^{|\mathcal{D}|\times d_1}$,
and $\mathcal{C}$ is a compact set in the domain of $\mbf{\gamma}$, such that $\mbf{\Sigma}(\mbf{\gamma})$ 
is a positive-definite matrix with eigenvalues bounded away from 0 and $\infty$. 
The true value of $\mbf{\Omega}$, denoted by $\mbf{\Omega}_0=\{\Lambda_{0jk}\}_{(j,k)\in\mathcal{D}}$, 
is continuously differentiable with positive derivatives in $[0,\tau]$.  
\end{condition} 

\begin{condition}\label{bvXZ}
With probability one, $\mbf{X}(t)$ and $\mbf{Z}(t)$ are continuously differentiable in $[0,\tau]$. 
If there exist a deterministic function $a_1(t)$ and a constant vector $\mbf{a}_2$ such that 
$a_1(t)+\mbf{a}_2^{\trans}\mbf{X}(t)=0$ with probability one, then $a_1(t)=0$ for $t\in[0,\tau]$ and $\mbf{a}_2=\mbf{0}$.
\end{condition}

\begin{condition} \label{S0}
The support of $S_0$ covers all the non-absorbing states among $1,\dots, K$, where an absorbing state is a state that cannot transition to any other state. 
\end{condition}

\begin{condition}\label{tau}
The number of examination times $N$ is positive with $E(N)<\infty$. 
The conditional probability $\pr(\tau_N=\tau\,|\,N, \mbf{X}, \mbf{Z})$ is greater than some positive constant $\eta_1$. In addition, with $[0,\tau]$ being the union of the supports of $(\tau_1,\dots,\tau_N)$, the conditional densities of $(\tau_{l-1},\tau_l)$ given $(N, \mbf{X}, \mbf{Z})$, denoted by $f_l(t_1,t_2)\; (l=1,\dots,N)$, have continuous second-order partial derivatives with respect to $t_1$ and $t_2$ when $t_2-t_1\geq\eta_2$ for some positive constant $\eta_2$, and are continuously differentiable functionals with respect to $\mbf{X}$ and $\mbf{Z}$. Finally, $\pr\{\min_{1\leq l\leq N}(\tau_{l}-\tau_{l-1})\geq\eta_2 \,|\,N, \mbf{X}, \mbf{Z}\}=1$.
\end{condition}

\begin{condition}\label{ident}
For a pair of parameters $(\mbf{\theta}_1, \mbf{\Omega}_1)$ and $(\mbf{\theta}_2, \mbf{\Omega}_2)$, if 
\[
\begin{split}
\int_{\mbf{b}} \mbf{P}(0,t; \mbf{X}, \mbf{Z}, \mbf{b},\mbf{\beta}_1, \mbf{\Omega}_1)\phi(\mbf{b};\,\mbf{\Sigma}(\mbf{\gamma}_1))d\mbf{b} 
= \int_{\mbf{b}}\mbf{P}(0,t; \mbf{X},\mbf{Z},\mbf{b},\mbf{\beta}_2, \mbf{\Omega}_2)\phi(\mbf{b};\,\mbf{\Sigma}(\mbf{\gamma}_2))d\mbf{b} 
\end{split}
\]
with probability one for any $t\in[0,\tau]$, then $\mbf{\beta}_1=\mbf{\beta}_2$, $\mbf{\gamma}_1=\mbf{\gamma}_2$, and $\mbf{\Omega}_1(t)=\mbf{\Omega}_2(t)$ for $t\in[0,\tau]$.
\end{condition}

\begin{condition}\label{nonsingular}
If there exist a $K\times K$ matrix-valued function $\mbf{a}_3(t;\mbf{b})$ and a $d_3$-vector $\mbf{a}_4$ such that
\[
\begin{split}
\int_{\mbf{b}}\Biggl[\int_0^t\mbf{P}(0,s;\mbf{X},\mbf{Z},\mbf{b},\mbf{\beta}_0,\mbf{\Omega}_0)d\mbf{a}_3(s;\mbf{b})\mbf{P}(s,t;\mbf{X},\mbf{Z},\mbf{b},\mbf{\beta}_0,\mbf{\Omega}_0)\Biggr. \\
\Biggl.+\mbf{P}(0,t;\mbf{X},\mbf{Z},\mbf{b},\mbf{\beta}_0,\mbf{\Omega}_0)\frac{{\mbf{a}}_4^{\trans}\phi^\prime_{\mbf{\gamma}}(\mbf{b};\mbf{\Sigma}(\mbf{\gamma}_0))}{\phi(\mbf{b};\mbf{\Sigma}(\mbf{\gamma}_0))}\Biggr]\phi(\mbf{b};\mbf{\Sigma}(\mbf{\gamma}_0))d\mbf{b}=\mbf{0}
\end{split}
\]
with probability one for any $t\in[0,\tau]$, where $\phi^\prime_{\mbf{\gamma}}$ is the derivative of $\phi(\mbf{b};\,\mbf{\Sigma}(\mbf{\gamma}))$ with respect to $\mbf{\gamma}$, then $\mbf{a}_3(t;\mbf{b})=\mbf{0}$ for $t\in[0,\tau]$ and $\mbf{a}_4=\mbf{0}$.
\end{condition}

\begin{remark}
Conditions~\ref{cond1} and \ref{bvXZ} are standard for regression analysis with time-dependent covariates. 
Condition~\ref{S0} assumes that the initial state can be any non-absorbing state, 
which ensures that all possible transitions can occur during the study. 
Condition~\ref{tau} pertains to the joint distribution of the examination times. 
First, it requires that the largest examination time reaches $\tau$ with positive probability. 
Second, it requires smoothness of the joint density of the examination times, 
which is used to prove the Donsker property of some function classes 
and the smoothness of the least favorable direction. 
Finally, this condition requires any two successive examination times to be separated by a positive gap; 
otherwise, transition times may be exactly observed, which calls for a different theoretical treatment. 
Conditions~\ref{ident} and \ref{nonsingular} ensure the identifiability of the proposed model 
and the invertibility of the information operator along any submodel under true parameter values. 
If $\mbf{X}$ and $\mbf{Z}$ are both time-independent, then Conditions~\ref{ident} and \ref{nonsingular} 
can be replaced by conditions
(1) $\mbf{Z}$ is linearly independent, that is, any symmetric matrix $\mbf{C}$ satisfying $\mbf{Z}^{\trans}\mbf{C}\mbf{Z}=0$ with probability one must be a zero matrix. 
(2) $\mbf{\Sigma}(\mbf{\gamma}_1) = \mbf{\Sigma}(\mbf{\gamma}_2)$ implies $\mbf{\gamma}_1 = \mbf{\gamma}_2$.
\end{remark}

We state the strong consistency of $(\widehat{\mbf{\theta}},\widehat{\mbf{\Omega}})$ 
and the limiting distribution of $n^{1/2}(\widehat{\mbf{\theta}}-\mbf{\theta}_0)$.

\begin{theorem} \label{thm1}
Under Conditions~\ref{cond1}-\ref{ident}, $\|\widehat{\mbf{\theta}}-\mbf{\theta}_0\|+\sum_{(j,k)\in\mathcal{D}}\|\widehat{\Lambda}_{jk}-\Lambda_{0jk}\|_{\infty}\rightarrow 0$ almost surely, where $\|\cdot\|$ is the Euclidean norm and $\|\cdot\|_{\infty}$ is the supremum norm over $[0,\tau]$.
\end{theorem}

\begin{theorem} \label{thm2}
Under Conditions~\ref{cond1}-\ref{nonsingular}, $n^{1/2}(\widehat{\mbf{\theta}}-\mbf{\theta}_0)$ converges in distribution to a multivariate normal vector with mean zero and covariance matrix that attains the semiparametric efficiency bound.
\end{theorem}

The proofs of the theorems are provided in the Appendix. 
The limiting covariance matrix of $\widehat{\mbf{\theta}}$ can be consistently estimated through
profile likelihood \citep{murphy2000profile}. 
Denote the profile log-likelihood for $\mbf{\theta}$ by 
$pl_n(\mbf{\theta}) = \max_{\mbf{\Omega}}\log L_n(\mbf{\theta}, \mbf{\Omega})$,
which can be obtained from the above EM algorithm with $\mbf{\theta}$ fixed. 
Let $pl_{ni}$ denote the $i$th subject's contribution to $pl_n$ 
and $\mbf{e}_j$ denote the $j$th canonical vector of the same dimension as $\mbf{\theta}$. 
Then the covariance matrix of $\widehat{\mbf{\theta}}$ can be estimated by the inverse of the matrix whose $(j,k)$th element is
$
\sum_{i=1}^n \{pl_{ni}(\widehat{\mbf{\theta}}+h_n\mbf{e}_j)-pl_{ni}(\widehat{\mbf{\theta}})\}\{pl_{ni}(\widehat{\mbf{\theta}}+h_n\mbf{e}_k)-pl_{ni}(\widehat{\mbf{\theta}})\}/h_n^2,
$
where $h_n$ is some constant of order $n^{-1/2}$.

\section{Simulation Studies}

We conducted a series of simulation studies with three states, which are numbered 1, 2, and 3. 
Possible transitions include 1 to 2 and 2 to 3. We generated two time-independent covariates, 
$X_1\sim\text{Ber}(0.5)$ and $X_2\sim \text{Unif}(0,1)$, and random effect $b\sim N(0,\sigma^2)$ with $\sigma^2=0.8$. 
We set $\Lambda_{12}(t) = \log(1+0.3t)$, $\Lambda_{23}(t) = 0.3t$, $(\beta_{121},\beta_{122}) = (0.5,-0.5)$,
$(\beta_{231},\beta_{232}) = (0.4, 0.2)$,
where $\Lambda_{jk}$ pertains to the transition from $j$ to $k$,
and $\beta_{jkl}$ pertains to the transition from $j$ to $k$ and the $l$th covariate.
The initial state of each subject was 1 or 2 with equal probabilities. 
We generated six potential examination times for each subject, 
with the first being $\text{Unif}(0,1)$, and the gap between any two successive examination times 
being $0.05+\text{Unif}(0,1)$. 
We set the study end time $\tau=3$ and excluded all the examinations beyond $\tau$. 
We simulated 10,000 replicates with $n=400$, 800, or 1600. 

We applied all three computational strategies described in Section~\ref{subsec: npmle}. 
We removed all the time points whose jump sizes were smaller than 0.0001.
We set the initial values of $\mbf{\beta}_{jk}$'s to $\mbf{0}$ and the initial values of $\lambda_{jks}$'s to $1/m$. 
In addition, we set the initial value of $\sigma^2$ to 1. 
The convergence criterion was that the maximal change in the parameter estimates at two successive iterations 
is smaller than $0.0001$. 
For the variance estimation, we set $h_n=5n^{-1/2}$, although the results differed only in the third decimal 
place when $h_n$ ranged from $n^{-1/2}$ to $10n^{-1/2}$.  

For comparisons, we included the \textbf{msm} package \citep{jackson2011multi}, which fits time-homogeneous 
or piecewise homogeneous Markov models. 
The implementation is via the \texttt{msm()} function, 
and the change points are specified in the \texttt{pci} argument 
when piecewise constant transition intensities are assumed. 
We let the function automatically generate the initial parameter values and 
used the default settings for maximum likelihood estimation. 
We placed the change points of the intensities at 0.5, 1, 1.5, 2, and 2.5.  

Table~\ref{sim1_res} summarizes the estimation results on the regression parameters. 
The EM algorithm converged in all replicates. 
The biases of the parameter estimators are small and decrease as $n$ increases. 
The variance estimators are accurate, and the confidence intervals have proper coverage probabilities. 
The parameter estimators in the \textbf{msm} package are severely biased.
When $n=400$, nearly 5\% of the replicates failed due to sparse data within some of the pieces.      

Figure~\ref{sim1_fig} shows the estimation results on the cumulative transition intensity functions. 
The median of the proposed estimates is almost identical to the truth, 
whereas the median of the estimates from \textbf{msm} deviates substantially from the truth. 


\begin{table}
\caption{Estimation of the regression parameters in the simulation studies with three states. \label{sim1_res}}
\begin{center}
\begin{tabular}{ll@{\hspace{1.5em}}*{4}{r}@{\hspace{1.5em}}*{4}{r}}
\toprule
  &  & \multicolumn{4}{c}{Proposed methods} & \multicolumn{4}{c}{\textbf{msm} package} \\
  \cmidrule(lr){3-6} \cmidrule(lr){7-10}
  & Parameter & Bias & SE & SEE & CP & Bias & SE & SEE & CP \\ 
  $n=400$ & $\beta_{121}=0.5$ & 0.014 & 0.265 & 0.259 & 95.0 & $-0.091$ & 0.209 & 0.207 & 92.4 \\ 
  & $\beta_{122}=-0.5$ & $-0.021$ & 0.458 & 0.448 & 94.7 & 0.087 & 0.363 & 0.356 & 94.0 \\ 
  & $\beta_{231}=0.4$ & 0.013 & 0.206 & 0.198 & 94.5 & $-0.078$ & 0.156 & 0.147 & 90.4 \\ 
  & $\beta_{232}=0.2$ & 0.005 & 0.350 & 0.339 & 94.5 & $-0.053$ & 0.268 & 0.254 & 92.8 \\ 
  & $\sigma^2=0.8$ & 0.060 & 0.422 & 0.396 & 95.1 &  &  &  &  \\ 
  $n=800$ & $\beta_{121}=0.5$ & 0.010 & 0.181 & 0.181 & 95.4 & $-0.092$ & 0.145 & 0.146 & 90.4 \\ 
  & $\beta_{122}=-0.5$ & $-0.008$ & 0.315 & 0.311 & 95.1 & 0.095 & 0.253 & 0.251 & 93.1 \\ 
  & $\beta_{231}=0.4$ & 0.007 & 0.139 & 0.138 & 95.3 & $-0.079$ & 0.107 & 0.104 & 87.4 \\ 
  & $\beta_{232}=0.2$ & 0.006 & 0.240 & 0.236 & 94.6 & $-0.053$ & 0.187 & 0.179 & 92.8 \\ 
  & $\sigma^2=0.8$ & 0.024 & 0.270 & 0.263 & 95.5 &  &  &  &  \\ 
  $n=1600$ & $\beta_{121}=0.5$ & 0.002 & 0.127 & 0.126 & 94.8 & -0.096 & 0.103 & 0.103 & 84.8 \\ 
  & $\beta_{122}=-0.5$ & $-0.000$ & 0.217 & 0.216 & 95.0 & 0.100 & 0.176 & 0.177 & 91.2 \\ 
  & $\beta_{231}=0.4$ & 0.000 & 0.098 & 0.096 & 94.9 & $-0.080$ & 0.076 & 0.073 & 79.7 \\ 
  & $\beta_{232}=0.2$ & $-0.002$ & 0.168 & 0.164 & 94.7 & $-0.057$ & 0.132 & 0.126 & 91.3 \\ 
  & $\sigma^2=0.8$ & $-0.004$ & 0.181 & 0.178 & 95.6 &  &  &  &  \\ 
  \bottomrule
    \end{tabular}
\end{center}
Note: Bias and SE denote the median bias and empirical standard error, respectively.
SEE denotes the median of the standard error estimator, and
CP denotes the empirical coverage percentage of the 95\% confidence interval. 
The log transformation was used to construct the confidence interval for $\sigma^2$. 
For \textbf{msm} with $n=400$, each entry is based on 9,490 replicates. 
All other entries are based on 10,000 replicates.
\end{table}

\begin{figure}
\begin{center}
\includegraphics[width=\textwidth]{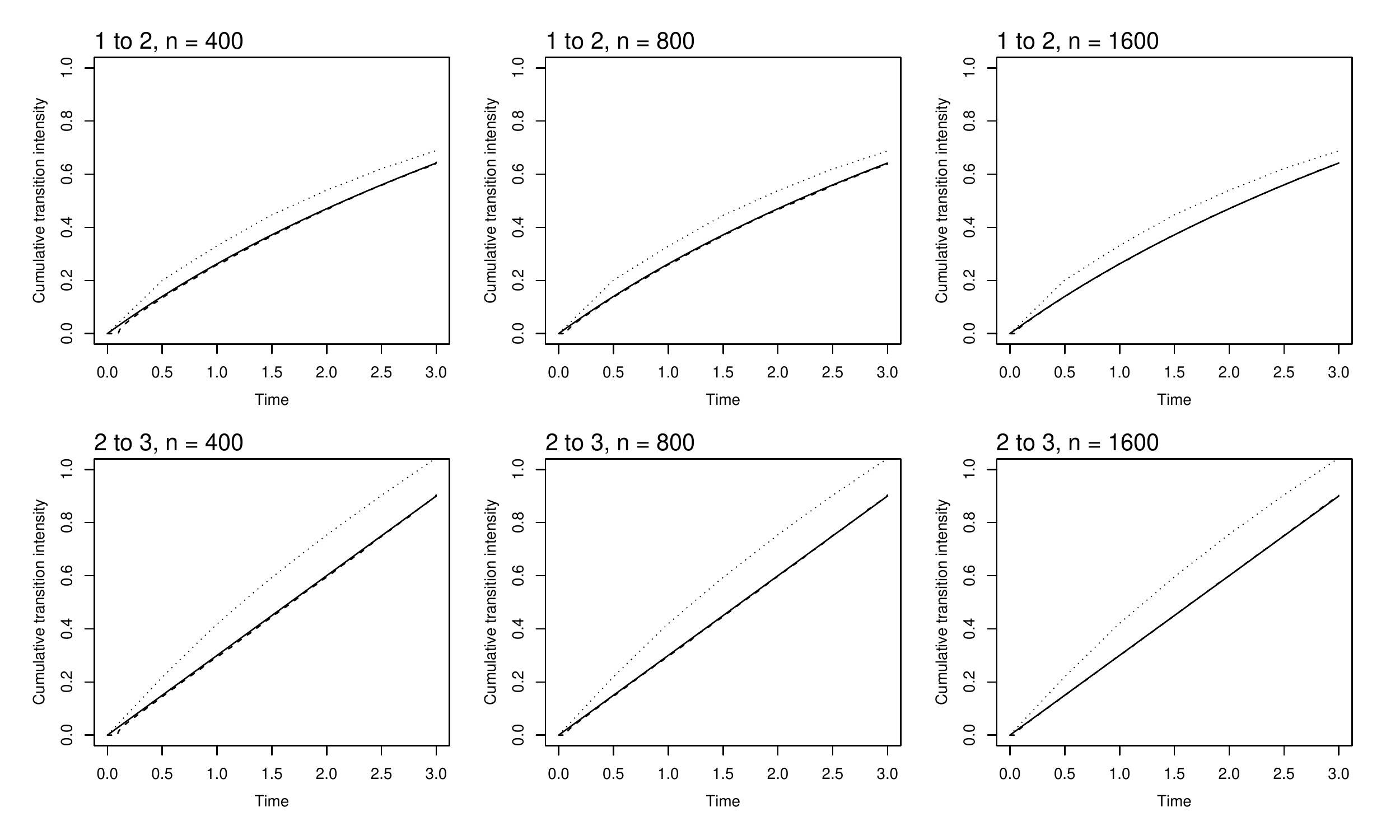}
\end{center}
\caption{Estimation of the cumulative transition intensities in the simulation studies with three states. 
The solid, dashed, and dotted curves show the true values, 
the median estimates based on 10,000 replicates using the proposed methods, 
and the median estimates based on 10,000 replicates (9,490 replicates for $n=400$) using the \textbf{msm} package, respectively. \label{sim1_fig}}
\end{figure}

Section S.2 of the supplementary materials reports simulation studies with more complex disease processes. 
The proposed methods continued to perform well.  

\section{Application}

The ARIC study recruited 15,792 participants aged 45--64 years in 1987--1989 from four communities: 
Forsyth County, North Carolina; Jackson, Mississippi; suburban Minneapolis, Minnesota; and Washington County, Maryland. 
All participants received a baseline examination upon enrollment, 
followed by three examinations conducted approximately every three years between 1990 and 1998, 
and three further examinations in 2011--2013, 2016--2017, and 2018--2019. 
At each of the last three examinations, MCI and dementia were assessed from current and longitudinal cognitive tests 
by a panel of reviewers (4 physicians and 4 neuropsychologists), 
yielding a syndromic diagnosis such that one of three states: normal, MCI, or dementia
was determined at each examination \citep{knopman2016mild}. 
It is unlikely for an individual to return to a less severe state from a more severe cognitive impairment state 
(e.g., MCI to normal, dementia to MCI). 

We considered a three-state progressive model (i.e., normal to MCI to dementia). 
The transitions between the three states were interval-censored. 
The time scale for the analysis was years since the baseline examination. 
We evaluated the effects of the following baseline risk factors on the transitions between the states:
age (years), gender (female vs. male), race-center (Forsyth County; Black, Jackson; White, Minneapolis; and White, Washington County), education level (basic or intermediate vs. advanced), diabetes (no vs. yes), cigarette smoking status (non-smoker vs. smoker), body mass index (kg/m$^2$), and systolic blood pressure (mmHg).  
We included a random intercept to capture the potential dependence between transitions. 
After removing participants with unknown states at the fifth examination or missing data on risk factors, 
a total of 6,407 participants remained. 
The mean follow-up time was 27.5 years and the median was 28.8 years. 
Table~\ref{statetable} summarizes the frequency that each pair of states was observed over successive examinations. 
The second, third, and fourth examinations are omitted, 
because no information about MCI or dementia was collected at those three examinations.

\begin{table}
\caption{Frequency for each pair of states over successive examinations among 6,407 participants in the ARIC study. \label{statetable}}
\begin{center}
\begin{tabular}{l*{3}{@{\hspace{3em}}r}}
\toprule
& \multicolumn{3}{c}{To} \\
\cmidrule(lr){2-4} 
{From} & Normal & MCI & Dementia \\
\midrule
Normal & 8,936 & 2,052 & 459 \\ 
MCI & 0 & 332 & 136 \\
Dementia & 0 & 0 & 214 \\
\bottomrule
\end{tabular}
\end{center}
\end{table}

We used the same estimation procedure as in the simulation studies, 
except that the threshold for jump sizes was set to $10^{-6}$. 
The number of unique time points was 3,155 in the beginning of the analysis and 147 at the end. 
The computation time was about two hours on a computer with Windows 10 (2.1 GHz processor, 32 GB RAM, 64-bit). 
The estimation results on the regression parameters are presented in Table~\ref{aric_res}. 
Older people have significantly higher risk of developing both MCI and dementia,
males are more likely to develop MCI, advanced education can significantly reduce the risk of progression from MCI to dementia,
people with diabetes have significantly higher risk of MCI, 
and baseline body mass index and systolic blood pressure are both positively associated with the risk of MCI. 
The variance of the random intercept was estimated at 0.9282, with estimated standard error of 0.1461, 
suggesting strong dependence between the transition from normal to MCI and the transition from MCI to dementia.  

\begin{table}
\caption{Estimation results on the regression parameters in the ARIC study. \label{aric_res}}
\begin{center}
\scalebox{0.9}{
\begin{tabular}{l*{3}{r}@{\hspace{2em}}*{3}{r}}
\toprule
 & \multicolumn{3}{c}{Normal to MCI} & \multicolumn{3}{c}{MCI to dementia} \\ 
 \cmidrule(lr){2-4} \cmidrule(lr){5-7} 
 Covariate & Estimate & St error & $p$-value & Estimate & St error & $p$-value \\
 \midrule
Age (years) & 0.0892 & 0.0031 & $<$0.0001 & 0.1110 & 0.0057 & $<$0.0001 \\ 
  Male & 0.3188 & 0.0520 & $<$0.0001 & 0.1636 & 0.1018 & 0.1080 \\ 
  Advanced education & $-0.1003$ & 0.0525 & 0.0561 & $-0.6164$ & 0.1097 & $<$0.0001 \\ 
  Diabetes & 0.5587 & 0.0994 & $<$0.0001 & 0.3962 & 0.1651 & 0.0164 \\ 
  Smoker & 0.1549 & 0.0661 & 0.0191 & 0.2009 & 0.1355 & 0.1382 \\ 
  Body mass index (kg/m$^2$) & 0.0213 & 0.0049 & $<$0.0001 & 0.0164 & 0.0090 & 0.0684 \\ 
  Systolic blood pressure (mmHg) & 0.0051 & 0.0015 & 0.0007 & 0.0058 & 0.0028 & 0.0383 \\ 
  Black, Jackson & $-0.0008$ & 0.0792 & 0.9919 & 1.4692 & 0.1614 & $<$0.0001 \\ 
  White, Minneapolis & $-0.2052$ & 0.0721 & 0.0044 & 0.4848 & 0.1621 & 0.0028 \\ 
  White, Washington County & $-0.0828$ & 0.0722 & 0.2515 & 0.5218 & 0.1593 & 0.0011 \\
  \bottomrule
  \end{tabular}}
\end{center}
Note: For each categorical variable, the group not shown is the reference group.
\end{table}

Figure~\ref{aric_Lambda} shows the estimated cumulative transition intensities for subjects with 
different combinations of education level and diabetes status and
with all other covariates set to be the sample medians. 
The left panel shows that having diabetes considerably increases the risk of MCI. 
The right panel shows that subjects with an advanced education have much lower risk of dementia 
than those without advanced education.  

\begin{figure}
\begin{center}
\includegraphics[width=\textwidth]{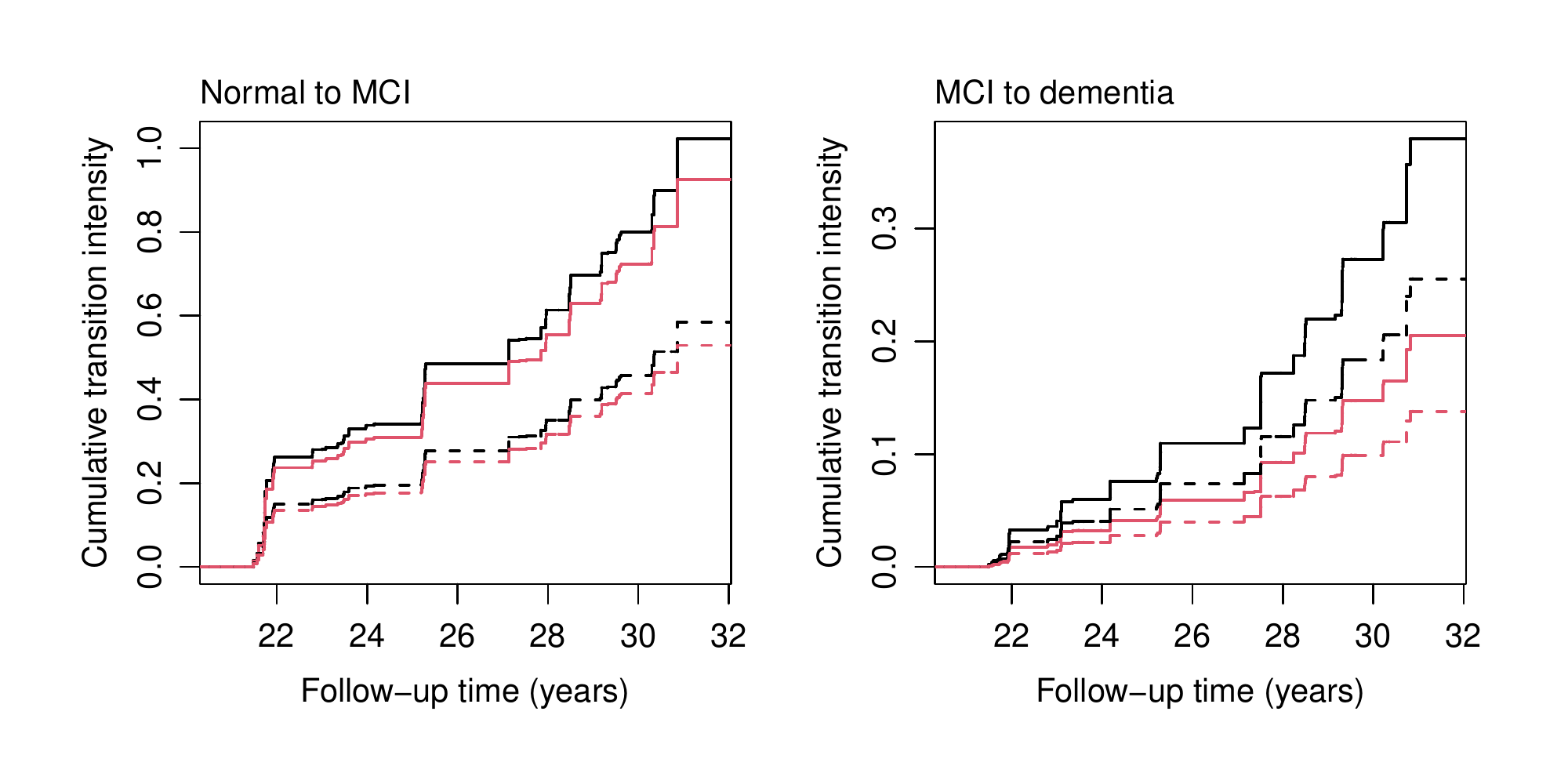}
\end{center}
\caption{Estimated cumulative transition intensities for subjects with different combinations
of education level and diabetes status at baseline in the ARIC study. 
The black curves pertain to subjects without advanced education, 
and the red curves pertain to subjects with advanced education. 
The solid curves pertain to subjects with diabetes, 
and the dashed curves pertain to subjects without diabetes. 
The other covariates are set to be the sample medians. \label{aric_Lambda}}
\end{figure}   

Figure~\ref{aric_prob} shows the estimated transition probabilities from normal and MCI to different states over five-year time 
intervals, with the covariates equal to the sample medians. 
Unsurprisingly, the probabilities of progression toward more severe states generally increase over time.   

\begin{figure}
\begin{center}
\includegraphics[width=.9\textwidth]{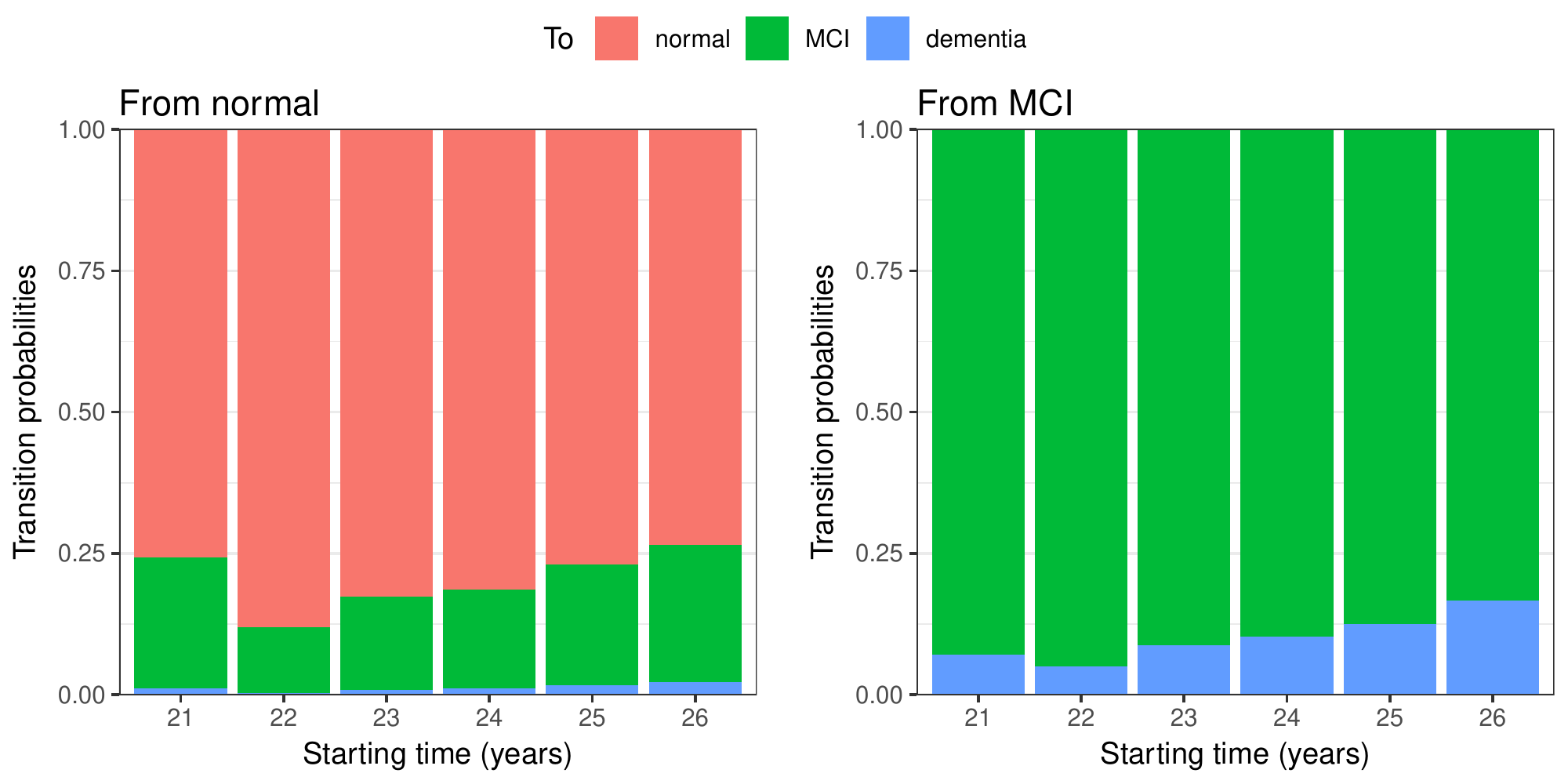}
\end{center}
\caption{Estimated transition probabilities over five-year time intervals for subjects
with median covariate values in the ARIC study. \label{aric_prob}}
\end{figure}

\section{Discussion}


We have developed powerful methods for analyzing very general interval-censored multi-state data. 
Unlike spline-based methods, we estimate the baseline transition intensity functions in a completely nonparametric manner 
and avoid any tuning parameters. 
We have established for the first time a rigorous asymptotic theory for the semiparametric estimation of multi-state models
under interval censorship. 
We have shown through extensive simulation studies that the proposed methods outperform the existing methods implemented 
in the \textbf{msm} package. 

Our work contains major innovations. 
First, the proposed EM algorithm is much more sophisticated and computationally challenging than that of
\citet{zeng2017maximum} because it is necessary to consider all possible transition paths 
when estimating conditional expectations. No such calculations were required in the case of 
multivariate interval-censored data. We also provide several strategies to significantly speed up the computation.   
Second, the presence of product integration poses substantial theoretical challenges, 
especially in proving the Donsker property of the relevant function classes and in handling the score and information operators. 
We have addressed these new challenges by using the results on product integration from \citet{andersen1993statistical}.      

Our formulation allows for an absorbing state and assumes that the transition time to the absorbing state is interval-censored. Sometimes the absorbing state can be exactly observed (e.g., death). Therefore, an interesting extension of our work is to study a mixture of interval- and right-censored data, where the transition times among the non-absorbing states are interval-censored, and the transition time to the absorbing state is right-censored. 

We have only considered routinely scheduled examinations and noninformative loss to follow-up. 
Inspired by the recent work of \citet{lawless2019new} and \citet{cook2021independence}, 
we may extend our work to allow for disease-driven examinations and informative loss to follow-up
by jointly modeling the disease process, the recurrent examination process, and the loss to follow-up process. 
Specifically, we may add loss to follow-up as a new state to the original state space, and we may consider a two-dimensional state space with one component characterizing the disease state and the other component counting the number of examinations. The former is relatively easy, while the latter can be very challenging unless strong assumptions about the transition intensities are made.

\section*{Acknowledgments}

The Atherosclerosis Risk in Communities Study is carried out as a collaborative study supported by National Heart, Lung, and Blood Institute contracts (75N92022D00001, 75N92022D00002, 75N92022D00003, 75N92022D00004, 75N92022D00005). The ARIC Neurocognitive Study is supported by U01HL096812, U01HL096814, U01HL096899, U01HL096902, and U01HL096917 from the NIH (NHLBI, NINDS, NIA and NIDCD). The authors thank the staff and participants of the ARIC study for their important contributions.
This research was supported by the National Institutes of Health grant R01HL149683.

\section*{Supplementary Materials}

The supplementary materials contain three lemmas and additional simulation results. 

\section*{Appendix. Proofs of Theorems}

The proofs of Theorems~\ref{thm1} and \ref{thm2} make use of three lemmas, which are stated and proved in Section S.1 of the supplementary materials. 
We use the notation: $\mathbb{P}_n$ denotes the empirical measure for $n$ independent subjects, $\mathbb{P}$ denotes the true probability measure, and $\mathbb{G}_n=n^{1/2}(\mathbb{P}_n-\mathbb{P})$ is the corresponding empirical process. 
Let $L(\mbf{\theta}, \mbf{\Omega})$ denote the likelihood function for a single subject
\[
L(\mbf{\theta}, \mbf{\Omega}) = \int_{\mbf{b}}\prod_{l=1}^{N}\mbf{P}(\tau_{l-1},\tau_{l}; \mbf{b}, \mbf{\beta}, \mbf{\Omega})^{(S_{l-1}, S_{l})}\phi(\mbf{b};\mbf{\Sigma}(\mbf{\gamma}))d\mbf{b},
\]
and let $\ell(\mbf{\theta}, \mbf{\Omega})$ denote the corresponding log-likelihood function.
For simplicity, we suppress the arguments $\mbf{X}$ and $\mbf{Z}$ in any transition probability matrix of the form $\mbf{P}(t_1,t_2; \mbf{X}, \mbf{Z}, \mbf{b}, \mbf{\beta}, \mbf{\Omega})$. 

\begin{proof}[Proof of Theorem~\ref{thm1}]
We first show that $\limsup_n\widehat{\Lambda}_{jk}(\tau)<\infty$ with probability one for any $(j,k)\in\mathcal{D}$. 
By the strong law of large numbers, $(\mathbb{P}_n-\mathbb{P})\ell(\mbf{\theta}_0, \mbf{\Omega}_0)\rightarrow0$ almost surely.
Then by the definition of the parameter estimators,
\[
\liminf_n\mathbb{P}_n\ell(\widehat{\mbf{\theta}}, \widehat{\mbf{\Omega}})\ge \liminf_n\mathbb{P}_n\ell(\mbf{\theta}_0, \mbf{\Omega}_0)=\mathbb{P}\ell(\mbf{\theta}_0, \mbf{\Omega}_0)
\]
with probability one.
In addition,
\begin{align*}
& \liminf_n\mathbb{P}_n\ell(\widehat{\mbf{\theta}}, \widehat{\mbf{\Omega}}) \\
={} & \liminf_n\mathbb{P}_n\log\left[\int_{\mbf{b}}\left\{\prod_{l=1}^N\mbf{P}(\tau_{l-1}, \tau_l;\mbf{b},\widehat{\mbf{\beta}}, \widehat{\mbf{\Omega}})^{(S_{l-1}, \,S_l)}\right\}\phi(\mbf{b};\mbf{\Sigma}(\widehat{\mbf{\gamma}}))d\mbf{b}\right] \\
\le{} & \liminf_n\mathbb{P}_nI(S_0=S_N)\log\left[\int_{\mbf{b}}\mbf{P}(0,\tau_N;\mbf{b},\widehat{\mbf{\beta}}, \widehat{\mbf{\Omega}})^{(S_0, \,S_0)}\phi(\mbf{b};\mbf{\Sigma}(\widehat{\mbf{\gamma}}))d\mbf{b}\right] \\
={} & 
\begin{multlined}[t]
\liminf_n\mathbb{P}_nI(S_0=S_N) \\
\times\log\Biggl[\int_{\mbf{b}}\exp\biggl\{-\sum_{k:\,(S_0,k)\in\mathcal{D}}\int_0^{\tau_N}\exp\left\{\mbf{\beta}_{S_0k}^{\trans}\mbf{X}(s)+\mbf{b}^{\trans}\mbf{Z}(s)\right\}d\widehat{\Lambda}_{S_0k}(s)\biggr\}\times\phi(\mbf{b};\mbf{\Sigma}(\widehat{\mbf{\gamma}}))d\mbf{b}\Biggr] 
\end{multlined} \\
\le{} & 
\begin{multlined}[t]
\liminf_n\mathbb{P}_nI(S_0=S_N,\, \tau_N=\tau) \\
\times\log\Biggl[\int_{\mbf{b}}\exp\biggl\{-\sum_{(S_0,k)\in\mathcal{D}}\int_0^{\tau}\exp\left\{\beta_{S_0k}^{\trans}\mbf{X}(s)+\mbf{b}^{\trans}\mbf{Z}(s)\right\}d\widehat{\Lambda}_{S_0k}(s)\biggr\}\times\phi(\mbf{b};\mbf{\Sigma}(\widehat{\mbf{\gamma}}))d\mbf{b}\Biggr]
\end{multlined} \\
\le{} &
\begin{multlined}[t]
\liminf_n\mathbb{P}_nI(S_0=S_N,\, \tau_N=\tau) \\
\times\log\Biggl[\int_{\mbf{b}}\exp\biggl\{-\sum_{(S_0,k)\in\mathcal{D}}\exp\left(-M-M\|\mbf{b}\|\right)\widehat{\Lambda}_{S_0k}(\tau)\biggr\}\times\phi(\mbf{b};\mbf{\Sigma}(\widehat{\mbf{\gamma}}))d\mbf{b}\Biggr],
\end{multlined} 
\end{align*}
where $M=\sup_{t\in[0,\tau]}\left\{\sup_{\mbf{X},\mbf{\beta}_{jk}}|\mbf{\beta}_{jk}^{\trans}\mbf{X}(t)|+\sup_{\mbf{Z}}|\mbf{Z}(t)|\right\}$ and is finite under Condition 2. 
Since for any $x>0$, $e^{-x}\le x^{-1}$, $\liminf_n\mathbb{P}_n\ell(\widehat{\mbf{\theta}}, \widehat{\mbf{\Omega}})$ can be further bounded from above by  
\begin{align*}
& 
\begin{multlined}[t]
\liminf_n\mathbb{P}_nI(S_0=S_N,\, \tau_N=\tau) \\
\times\log\Biggl[\int_b\biggl\{\exp\left(-M-M\|b\|\right)\times\sum_{(S_0,k)\in\mathcal{D}}\widehat{\Lambda}_{S_0k}(\tau)\biggr\}^{-1}\times\phi(\mbf{b};\mbf{\Sigma}(\widehat{\mbf{\gamma}}))d\mbf{b}\Biggr]
\end{multlined} \\
\le{} & \liminf_n\mathbb{P}_nI(S_0=S_N,\, \tau_N=\tau)\Biggl[C(M)-\log\biggl\{\sum_{(S_0,k)\in\mathcal{D}}\widehat{\Lambda}_{S_0k}(\tau)\biggr\}\Biggr],
\end{align*}
where $C(M)$ is a deterministic function of $M$. 
Under Condition 4, $\lim_n\mathbb{P}_nI(S_0=S_N,\, \tau_N=\tau) = \pr(S_0=S_N,\, \tau_N=\tau)>0$, such that
\[
\limsup_n\log\biggl\{\sum_{(S_0,k)\in\mathcal{D}}\widehat{\Lambda}_{S_0k}(\tau)\biggr\}\le C(M)-O(1)\times\mathbb{P}\ell(\mbf{\theta}_0, \mbf{\Omega}_0)<\infty
\]
with probability one. Since $S_0$ can take an arbitrary value in $\{1,\dots, K\}$ under Condition~\ref{S0}, the above inequality implies that $\limsup_n\widehat{\Lambda}_{jk}(\tau)\le w <\infty$ with probability one for some positive finite constant $w$ and any $(j,k)\in\mathcal{D}$. 

We have shown that each component of $\widehat{\mbf{\Omega}}$ has bounded total variation in $[0,\tau]$. By Helly's selection lemma, for any subsequence of $(\widehat{\mbf{\theta}}, \widehat{\mbf{\Omega}})$, we can choose a further subsequence such that $\widehat{\mbf{\Omega}}$ converges to $\mbf{\Omega}^*$ pointwise in $[0,\tau]$, and that $\widehat{\mbf{\theta}}=(\widehat{\mbf{\beta}}, \widehat{\mbf{\gamma}})$ converges to $\mbf{\theta}^*=(\mbf{\beta}^*, \mbf{\gamma}^*)$. Next, we will show that $(\mbf{\theta}^*, \mbf{\Omega}^*)=(\mbf{\theta}_0, \mbf{\Omega}_0)$. 
Define the function
\[
m(\mbf{\theta}, \mbf{\Omega})=\log\left\{\frac{L(\mbf{\theta}, \mbf{\Omega})+L(\mbf{\theta}_0, \mbf{\Omega}_0)}{2}\right\}
\]
and class
\[
\mathcal{M} = \left\{m(\mbf{\theta}, \mbf{\Omega}):\; \mbf{\theta}\in\Theta, \mbf{\Omega}\in\mathcal{L}_w\right\},
\]
where $\mathcal{L}_w$ is the set of $|\mathcal{D}|$-dimensional non-decreasing functions $\{\Lambda_{jk}\}_{(j,k)\in\mathcal{D}}$ whose
total variations in $[0,\tau]$ are bounded by $w$, with $\Lambda_{jk}(0) = 0$. 
By the concavity of the log function, 
\[
\mathbb{P}_nm(\widehat{\mbf{\theta}}, \widehat{\mbf{\Omega}})\ge\mathbb{P}_n\frac{\ell(\widehat{\mbf{\theta}}, \widehat{\mbf{\Omega}})+\ell(\mbf{\theta}_0, \mbf{\Omega}_0)}{2}\ge\mathbb{P}_n\ell(\mbf{\theta}_0, \mbf{\Omega}_0)=\mathbb{P}_nm(\mbf{\theta}_0, \mbf{\Omega}_0),
\]
which implies 
\begin{equation}\label{eq1}
(\mathbb{P}_n-\mathbb{P})m(\widehat{\mbf{\theta}}, \widehat{\mbf{\Omega}})+\mathbb{P}m(\widehat{\mbf{\theta}}, \widehat{\mbf{\Omega}})\ge(\mathbb{P}_n-\mathbb{P})m(\mbf{\theta}_0, \mbf{\Omega}_0)+\mathbb{P}m(\mbf{\theta}_0, \mbf{\Omega}_0).
\end{equation}
We show in Lemma S.1 that $\mathcal{M}$ is a Donsker class, and we have verified that $m(\widehat{\mbf{\theta}}, \widehat{\mbf{\Omega}})\in\mathcal{M}$. Thus, $(\mathbb{P}_n-\mathbb{P})m(\widehat{\mbf{\theta}}, \widehat{\mbf{\Omega}})\rightarrow0$ almost surely.
In addition, $(\mathbb{P}_n-\mathbb{P})m(\mbf{\theta}_0, \mbf{\Omega}_0) = (\mathbb{P}_n-\mathbb{P})\ell(\mbf{\theta}_0, \mbf{\Omega}_0) \rightarrow0$ almost surely.
Because $\left|\prod_{l=1}^N\mbf{P}(\tau_{l-1}, \tau_l;\mbf{b},\mbf{\beta}, \mbf{\Omega})^{(S_{l-1}, \,S_l)}\right|<1$ for any $\mbf{\beta}\in\mathcal{B}$ and $\mbf{\Omega}\in\mathcal{L}_w$ with probability one, we conclude that with respect to the probability measure for $(\tau_1,\dots,\tau_N)$, 
\[
\int_{\mbf{b}}\left\{\prod_{l=1}^N\mbf{P}(\tau_{l-1}, \tau_l;\mbf{b},\mbf{\beta}^*, \widehat{\mbf{\Omega}})^{(S_{l-1}, \,S_l)}-\prod_{l=1}^N\mbf{P}(\tau_{l-1}, \tau_l;\mbf{b},\mbf{\beta}^*, \mbf{\Omega}^*)^{(S_{l-1}, \,S_l)}\right\}\phi(\mbf{b};\mbf{\Sigma}(\mbf{\gamma}^*))d\mbf{b}\rightarrow0.
\]
By the dominated convergence theorem and the fact that $L(\mbf{\theta}, \mbf{\Omega})+L(\mbf{\theta}_0, \mbf{\Omega}_0)$ is bounded away from zero for any $\mbf{\theta}\in\Theta$ and $\mbf{\Omega}\in\mathcal{L}_w$,  
\begin{align*}
& \left|\mathbb{P}m(\widehat{\mbf{\theta}}, \widehat{\mbf{\Omega}})-\mathbb{P}m(\mbf{\theta}^*, \mbf{\Omega}^*)\right| \\
\le{} & \left|\mathbb{P}m(\widehat{\mbf{\theta}}, \widehat{\mbf{\Omega}})-\mathbb{P}m(\mbf{\theta}^*, \widehat{\mbf{\Omega}})\right|+ \left|\mathbb{P}m(\mbf{\theta}^*, \widehat{\mbf{\Omega}})-\mathbb{P}m(\mbf{\theta}^*, \mbf{\Omega}^*)\right| \\
={} & O(\|\widehat{\mbf{\theta}}-\mbf{\theta}^*\|)+\mathbb{P}\log\frac{\int_{\mbf{b}}\left\{\prod_{l=1}^N\mbf{P}(\tau_{l-1}, \tau_l;\mbf{b}, \mbf{\beta}^*, \widehat{\mbf{\Omega}})^{(S_{l-1}, \,S_l)}\right\}\phi(\mbf{b};\mbf{\Sigma}(\mbf{\gamma}^*))d\mbf{b}+L(\mbf{\theta}_0, \mbf{\Omega}_0)}{\int_{\mbf{b}}\left\{\prod_{l=1}^N\mbf{P}(\tau_{l-1}, \tau_l;\mbf{b},\mbf{\beta}^*, \mbf{\Omega}^*)^{(S_{l-1}, \,S_l)}\right\}\phi(\mbf{b};\mbf{\Sigma}(\mbf{\gamma}^*))d\mbf{b}+L(\mbf{\theta}_0, \mbf{\Omega}_0)} \\
\rightarrow{} & 0.
\end{align*}
Thus, $\mathbb{P}m(\widehat{\mbf{\theta}}, \widehat{\mbf{\Omega}})\rightarrow\mathbb{P}m(\mbf{\theta}^*, \mbf{\Omega}^*)$ almost surely. Now we can take the limits on both sides of \eqref{eq1} and finally obtain $\mathbb{P}m(\mbf{\theta}^*, \mbf{\Omega}^*)\ge\mathbb{P}m(\mbf{\theta}_0, \mbf{\Omega}_0)$. By the properties of the Kullback-Leibler information, $L(\mbf{\theta}^*, \mbf{\Omega}^*)=L(\mbf{\theta}_0, \mbf{\Omega}_0)$ with probability one. Therefore,
\[
\begin{split}
\int_{\mbf{b}}\left\{\prod_{l=1}^N\mbf{P}(\tau_{l-1}, \tau_l;\mbf{b},\mbf{\beta}^*, \mbf{\Omega}^*)^{(S_{l-1}, \,S_l)}\right\}\phi(\mbf{b};\mbf{\Sigma}(\mbf{\gamma}^*))d\mbf{b} \\
=\int_{\mbf{b}}\left\{\prod_{l=1}^N\mbf{P}(\tau_{l-1}, \tau_l;\mbf{b},\mbf{\beta}_0, \mbf{\Omega}_0)^{(S_{l-1}, \,S_l)}\right\}\phi(\mbf{b};\mbf{\Sigma}(\mbf{\gamma}_0))d\mbf{b}.
\end{split}
\]
For any fixed sequence of monitoring times $0=\tau_0<\tau_1<\cdots<\tau_N\le\tau$, and any feasible start and end states $(S_0,\,S_N)$, we let $(S_1,S_2,\dots,S_{N-1})$ go over all possible combinations. Then the summation of the resulting equations yields
\[
\begin{split}
\int_{\mbf{b}}\left\{\sum_{(S_1,\dots,S_{N-1})}\prod_{l=1}^N\mbf{P}(\tau_{l-1}, \tau_l;\mbf{b},\mbf{\beta}^*, \mbf{\Omega}^*)^{(S_{l-1}, \,S_l)}\right\}\phi(\mbf{b};\mbf{\Sigma}(\mbf{\gamma}^*))d\mbf{b} \\
=\int_{\mbf{b}}\left\{\sum_{(S_1,\dots,S_{N-1})}\prod_{l=1}^N\mbf{P}(\tau_{l-1}, \tau_l;\mbf{b},\mbf{\beta}_0, \mbf{\Omega}_0)^{(S_{l-1}, \,S_l)}\right\}\phi(\mbf{b};\mbf{\Sigma}(\mbf{\gamma}_0))d\mbf{b},
\end{split}
\]
which implies 
\[
\int_{\mbf{b}}\mbf{P}(0, \tau_N;\mbf{b},\mbf{\beta}^*, \mbf{\Omega}^*)^{(S_0, \,S_N)}\phi(\mbf{b};\mbf{\Sigma}(\mbf{\gamma}^*))d\mbf{b}=\int_{\mbf{b}}\mbf{P}(0, \tau_N;\mbf{b},\mbf{\beta}_0, \mbf{\Omega}_0)^{(S_0, \,S_N)}\phi(\mbf{b};\mbf{\Sigma}(\mbf{\gamma}_0))d\mbf{b}.
\]
The above equation holds for any $\tau_N\in[0,\tau]$ and any feasible $(S_0, S_N)$, which covers the whole set $\mathcal{D}$ under Condition~\ref{S0}. Thus, for any $t\in[0,\tau]$,
\[
\int_{\mbf{b}}\mbf{P}(0, t;\mbf{b},\mbf{\beta}^*, \mbf{\Omega}^*)\phi(\mbf{b};\mbf{\Sigma}(\mbf{\gamma}^*))d\mbf{b}=\int_{\mbf{b}}\mbf{P}(0, t;\mbf{b},\mbf{\beta}_0, \mbf{\Omega}_0)\phi(\mbf{b};\mbf{\Sigma}(\mbf{\gamma}_0))d\mbf{b},
\]
with probability one. By the identifiability in Condition 5, $\mbf{\beta}^*=\mbf{\beta}_0$, $\mbf{\gamma}^*=\mbf{\gamma}_0$, and $\mbf{\Omega}^*(t)=\mbf{\Omega}_0(t)$ for $t\in[0,\tau]$. The continuity of $\mbf{\Omega}_0(t)$ further implies $\|\widehat{\mbf{\theta}}-\mbf{\theta}_0\|+\sum_{(j,k)\in\mathcal{D}}\|\widehat{\Lambda}_{jk}-\Lambda_{0jk}\|_{\infty}\rightarrow 0$ almost surely. 
\end{proof}

\begin{proof}[Proof of Theorem~\ref{thm2}]
For any $l\in\{1,\dots,N\}$, $(j,k)\in\mathcal{D}$, and $t\in[0,\tau]$, let 
\[
B_{ljk}(t;\mbf{\theta},\mbf{\Omega})=
\begin{multlined}[t]
\frac{1}{L(\mbf{\theta},\mbf{\Omega})}\int_{\mbf{b}}\left\{\prod_{l'\ne l}\mbf{P}(\tau_{l'-1}, \tau_{l'};\mbf{b},\mbf{\beta}, \mbf{\Omega})^{(S_{l'-1}, \,S_{l'})}\right\} \\
\times \mbf{P}(\tau_{l-1}, t;\mbf{b},\mbf{\beta},\mbf{\Omega})^{(S_{l-1},\,j)}\exp\left\{\mbf{\beta}_{jk}^{\trans}\mbf{X}(t)+\mbf{b}^{\trans}\mbf{Z}(t)\right\} \\
\times \left\{\mbf{P}(t,\tau_{l};\mbf{b},\mbf{\beta},\mbf{\Omega})^{(k,\,S_{l})}-\mbf{P}(t,\tau_{l};\mbf{b},\mbf{\beta},\mbf{\Omega})^{(j,\,S_{l})}\right\} \\
\times I(\tau_{l-1}<t\le\tau_l)\phi(\mbf{b};\mbf{\Sigma}(\mbf{\gamma}))d\mbf{b}. 
\end{multlined}
\]
The score function for $\mbf{\theta}$ is 
\[
\mbf{\ell}_{\mbf{\theta}}(\mbf{\theta},\mbf{\Omega})=
\begin{bmatrix}
\mbf{\ell}_{\mbf{\beta}}(\mbf{\theta},\mbf{\Omega}) \\
\mbf{\ell}_{\mbf{\gamma}}(\mbf{\theta},\mbf{\Omega}) \\
\end{bmatrix},
\]
where $\mbf{\ell}_{\mbf{\beta}}(\mbf{\theta},\mbf{\Omega}) = \{\mbf{\ell}_{\mbf{\beta}_{jk}}(\mbf{\theta},\mbf{\Omega})\}_{(j,k)\in\mathcal{D}}$,
\begin{gather*}
\mbf{\ell}_{\mbf{\beta}_{jk}}(\mbf{\theta},\mbf{\Omega})=\sum_{l=1}^N\int_0^\tau B_{ljk}(t;\mbf{\theta},\mbf{\Omega})\mbf{X}(t)d\Lambda_{jk}(t), \\
\mbf{\ell}_{\mbf{\gamma}}(\mbf{\theta},\mbf{\Omega})=\frac{1}{L(\mbf{\theta},\mbf{\Omega})}\int_{\mbf{b}} \left\{\prod_{l=1}^N\mbf{P}(\tau_{l-1}, \tau_l;\mbf{b},\mbf{\beta}, \mbf{\Omega})^{(S_{l-1}, \,S_l)}\right\}\phi^{\prime}_{\mbf{\gamma}}(\mbf{b};\mbf{\Sigma}(\mbf{\gamma}))d\mbf{b}.
\end{gather*}
To obtain the score operator for $\mbf{\Omega}$, we consider a one-dimensional submodel $d\mbf{\Omega}_{\epsilon, \mbf{h}} = \{d\Lambda_{jk,\epsilon, h_{jk}}\}_{(j,k)\in\mathcal{D}}$ with each component defined by $d\Lambda_{jk,\epsilon, h_{jk}}(t)=(1+\epsilon h_{jk}(t))d\Lambda_{jk}(t)$, where $\mbf{h}=\{h_{jk}\}_{(j,k)\in\mathcal{D}}\in \mathcal{H} = \prod_{(j,k)\in\mathcal{D}}L_2(\mu_{jk})$ with $\mu_{jk}$ being the measure generated by $\Lambda_{jk}$. Under the true values $\Lambda_{0jk}$ ($(j,k)\in\mathcal{D}$), $\mathcal{H}$ is equivalent to the space $L_2([0,\tau])^{|\mathcal{D}|}$, since $\Lambda_{0jk}$ is continuously differentiable in $[0,\tau]$. The score function for $\mbf{\Omega}$ along this submodel is 
\[
\ell_{\mbf{\Omega}}(\mbf{\theta},\mbf{\Omega})(\mbf{h}) = \frac{\partial}{\partial \epsilon}\ell(\mbf{\theta},\mbf{\Omega}_{\epsilon, \mbf{h}})\bigg\vert_{\epsilon=0}=\sum_{l=1}^N\sum_{(j,k)\in\mathcal{D}}\int_0^\tau B_{ljk}(t;\mbf{\theta},\mbf{\Omega})h_{jk}(t)d\Lambda_{jk}(t).
\]
By the definition of the parameter estimators, $\mathbb{P}_n\{\mbf{\ell}_{\mbf{\theta}}(\widehat{\mbf{\theta}},\widehat{\mbf{\Omega}})\}=\mbf{0}$ and $\mathbb{P}_n\{\ell_{\mbf{\Omega}}(\widehat{\mbf{\theta}},\widehat{\mbf{\Omega}})(\mbf{h})\}=0$. 
In addition, $\mathbb{P}\{\mbf{\ell}_{\mbf{\theta}}(\mbf{\theta}_0,\mbf{\Omega}_0)\}=\mbf{0}$ and $\mathbb{P}\{\ell_{\mbf{\Omega}}(\mbf{\theta}_0,\mbf{\Omega}_0)(\mbf{h})\}=0$. Hence,
\begin{gather*}
\mathbb{G}_n\{\mbf{\ell}_{\mbf{\theta}}(\widehat{\mbf{\theta}},\widehat{\mbf{\Omega}})\}=-n^{1/2}\left[\mathbb{P}\{\mbf{\ell}_{\mbf{\theta}}(\widehat{\mbf{\theta}},\widehat{\mbf{\Omega}})\}-\mathbb{P}\{\mbf{\ell}_{\mbf{\theta}}(\mbf{\theta}_0,\mbf{\Omega}_0)\}\right], \\
\mathbb{G}_n\{\ell_{\mbf{\Omega}}(\widehat{\mbf{\theta}},\widehat{\mbf{\Omega}})(\mbf{h})\}=-n^{1/2}\left[\mathbb{P}\{\ell_{\mbf{\Omega}}(\widehat{\mbf{\theta}},\widehat{\mbf{\Omega}})(\mbf{h})\}-\mathbb{P}\{\ell_{\mbf{\Omega}}(\mbf{\theta}_0,\mbf{\Omega}_0)(\mbf{h})\}\right].
\end{gather*}
We apply Taylor expansion at $(\mbf{\theta}_0,\mbf{\Omega}_0)$ to the right-hand sides of the above two equations. By Lemma S.2, the second-order terms are bounded by
\begin{align*}
& n^{1/2}\Biggl\{O(1)E\biggl[\sum_{(j,k)\in\mathcal{D}}\sum_{l=1}^N\left\{\widehat{\Lambda}_{jk}(\tau_l)-\Lambda_{0jk}(\tau_l)\right\}^2\biggr]+O(1)\|\widehat{\mbf{\beta}}-\mbf{\beta}_0\|^2+O(1)\|\widehat{\mbf{\gamma}}-\mbf{\gamma}_0\|^2\Biggr\} \\
\le{} & n^{1/2}\left\{O_p\left(n^{-2/3}\right)+O_p\left(\|\widehat{\mbf{\beta}}-\mbf{\beta}_0\|^2+\|\widehat{\mbf{\gamma}}-\mbf{\gamma}_0\|^2\right)\right\} \\
={} & O_p\left(n^{1/2}\|\widehat{\mbf{\beta}}-\mbf{\beta}_0\|^2+n^{1/2}\|\widehat{\mbf{\gamma}}-\mbf{\gamma}_0\|^2+n^{-1/6}\right).
\end{align*}
Therefore, 
\begin{gather}
\label{g1}
\mathbb{G}_n\{\mbf{\ell}_{\mbf{\theta}}(\widehat{\mbf{\theta}},\widehat{\mbf{\Omega}})\}=
\begin{multlined}[t]
-n^{1/2}\mathbb{P}\{\mbf{\ell}_{\mbf{\theta}\mbf{\theta}}(\widehat{\mbf{\theta}}-\mbf{\theta}_0)+\mbf{\ell}_{\mbf{\theta}\mbf{\Omega}}(\widehat{\mbf{\Omega}}-\mbf{\Omega}_0)\} \\
+O_p\left(n^{1/2}\|\widehat{\mbf{\beta}}-\mbf{\beta}_0\|^2+n^{1/2}\|\widehat{\mbf{\gamma}}-\mbf{\gamma}_0\|^2+n^{-1/6}\right), 
\end{multlined} \\ 
\label{g2}
\mathbb{G}_n\{\ell_{\mbf{\Omega}}(\widehat{\mbf{\theta}},\widehat{\mbf{\Omega}})(\mbf{h})\}=
\begin{multlined}[t]
-n^{1/2}\mathbb{P}\{\mbf{\ell}_{\mbf{\Omega}\mbf{\theta}}(\mbf{h})^{\trans}(\widehat{\mbf{\theta}}-\mbf{\theta}_0)+\ell_{\mbf{\Omega}\mbf{\Omega}}(\mbf{h},\widehat{\mbf{\Omega}}-\mbf{\Omega}_0)\} \\
+O_p\left(n^{1/2}\|\widehat{\mbf{\beta}}-\mbf{\beta}_0\|^2+n^{1/2}\|\widehat{\mbf{\gamma}}-\mbf{\gamma}_0\|^2+n^{-1/6}\right),
\end{multlined} 
\end{gather}
where $\mbf{\ell}_{\mbf{\theta}\mbf{\theta}}$ is the second derivative of $\ell(\mbf{\theta},\mbf{\Omega})$ with respect to $\mbf{\theta}$, $\mbf{\ell}_{\mbf{\theta}\mbf{\Omega}}(\mbf{h})$ is the derivative of $\mbf{\ell}_{\mbf{\theta}}$ along the submodel $d\mbf{\Omega}_{\epsilon, \mbf{h}}$, $\mbf{\ell}_{\mbf{\Omega}\mbf{\theta}}(\mbf{h})$ is the derivative of $\ell_{\mbf{\Omega}}(\mbf{h})$ with respect to $\mbf{\theta}$, and $\ell_{\mbf{\Omega}\mbf{\Omega}}(\mbf{h},\widehat{\mbf{\Omega}}-\mbf{\Omega}_0)$ is the derivative of $\ell_{\mbf{\Omega}}(\mbf{h})$ along the submodel $d\mbf{\Omega}_0+\epsilon d(\widehat{\mbf{\Omega}}-\mbf{\Omega}_0) = \{d\Lambda_{0jk}+\epsilon d(\widehat{\Lambda}_{jk}-\Lambda_{0jk})\}_{(j,k)\in\mathcal{D}}$. All the derivatives are evaluated at $(\mbf{\theta}_0,\mbf{\Omega}_0)$.

Let $\ell_{\mbf{\Omega}}^*:\,L_2(\mathbb{P})\rightarrow \mathcal{H}$ be the adjoint operator of $\ell_{\mbf{\Omega}}$. 
We define ${\mbf{h}}^*$ to be the least favorable direction such that 
$\ell_{\mbf{\Omega}}^*\ell_{\mbf{\Omega}}({\mbf{h}}^*)=\ell_{\mbf{\Omega}}^*\mbf{\ell}_{\mbf{\theta}}$. 
Lemma S.3 establishes the existence of $\mbf{h}^*$. Note that ${\mbf{h}}^*$ is a $(|\mathcal{D}|\times d_1+d_3)$-dimensional vector of functions in $\mathcal{H}$. Thus,
\begin{align*}
E\{\ell_{\mbf{\Omega}\mbf{\Omega}}({\mbf{h}}^*,\widehat{\mbf{\Omega}}-\mbf{\Omega}_0)\}={} & -E\left\{\ell_{\mbf{\Omega}}({\mbf{h}}^*)\ell_{\mbf{\Omega}}(\widehat{\mbf{\Omega}}-\mbf{\Omega}_0)\right\} \\
={} & -\int \ell_{\mbf{\Omega}}^*\ell_{\mbf{\Omega}}({\mbf{h}}^*)(d\widehat{\mbf{\Omega}}-d\mbf{\Omega}_0)=-\int \ell_{\mbf{\Omega}}^*\mbf{\ell}_{\mbf{\theta}}(d\widehat{\mbf{\Omega}}-d\mbf{\Omega}_0) \\
={} & -E\left\{\mbf{\ell}_{\mbf{\theta}}\ell_{\mbf{\Omega}}(\widehat{\mbf{\Omega}}-\mbf{\Omega}_0)\right\}=E\left\{\mbf{\ell}_{\mbf{\theta}\mbf{\Omega}}(\widehat{\mbf{\Omega}}-\mbf{\Omega}_0)\right\},
\end{align*}
so the difference between \eqref{g1} and \eqref{g2} yields 
\[
\mathbb{G}_n\left\{\mbf{\ell}_{\mbf{\theta}}(\widehat{\mbf{\theta}},\widehat{\mbf{\Omega}})-\ell_{\mbf{\Omega}}(\widehat{\mbf{\theta}},\widehat{\mbf{\Omega}})({\mbf{h}}^*)\right\}=
\begin{multlined}[t]
n^{1/2}E\left[\left\{\mbf{\ell}_{\mbf{\theta}}-\ell_{\mbf{\Omega}}({\mbf{h}}^*)\right\}^{\otimes2}\right](\widehat{\mbf{\theta}}-\mbf{\theta}_0) \\
+O_p\left(n^{1/2}\|\widehat{\mbf{\beta}}-\mbf{\beta}_0\|^2+n^{1/2}\|\widehat{\mbf{\gamma}}-\mbf{\gamma}_0\|^2+n^{-1/6}\right). 
\end{multlined}
\]
By Lemma S.3, the above equation entails that $n^{1/2}(\widehat{\mbf{\theta}}-\mbf{\theta}_0)=O_p(1)$ and yields
\[
n^{1/2}(\widehat{\mbf{\theta}}-\mbf{\theta}_0)=\left(E\left[\left\{\mbf{\ell}_{\mbf{\theta}}-\ell_{\mbf{\Omega}}({\mbf{h}}^*)\right\}^{\otimes2}\right]\right)^{-1}\mathbb{G}_n\{\mbf{\ell}_{\mbf{\theta}}-\ell_{\mbf{\Omega}}({\mbf{h}}^*)\}+o_p(1).
\]
This implies that the influence function for $\widehat{\mbf{\theta}}$ is exactly the efficient influence function, such that $n^{1/2}(\widehat{\mbf{\theta}}-\mbf{\theta}_0)$ converges weakly to a zero-mean multivariate normal vector whose covariance matrix attains the semiparametric efficiency bound \citep{bickel1993efficient}.
\end{proof}

\bibliography{Bibliography-MM-MC}
\end{document}


\bibliographystyle{agsm}

\def\spacingset#1{\renewcommand{\baselinestretch}%
{#1}\small\normalsize} \spacingset{1}


\if1\blind
{
  \title{\bf Supplementary Materials for ``Maximum Likelihood Estimation for Semiparametric Regression Models with Multi-State Interval-Censored Data''}
  \author{Yu Gu, Donglin Zeng, Gerardo Heiss, and D. Y. Lin}
  \date{\vspace{-5ex}}
  \maketitle
} \fi

\if0\blind
{
  \bigskip
  \bigskip
  \bigskip
  \begin{center}
    {\Large\bf Supplementary Materials for ``Maximum Likelihood Estimation for Semiparametric Regression Models with Multi-State Interval-Censored Data''}
\end{center}
  \medskip
} \fi

\bigskip
\spacingset{1.9} 

\section{Some Useful Lemmas}

\begin{lemma} \label{lemma: Donsker}
Under Conditions 1-5, the class $\mathcal{M}$ is Donsker.
\end{lemma}
\begin{proof}
Let 
\begin{gather*}
M=\sup_{t\in[0,\tau]}\left\{\sup_{\mbf{X},\mbf{\beta}_{jk}}|\mbf{\beta}_{jk}^{\trans}\mbf{X}(t)|+\sup_{\mbf{Z}}|\mbf{Z}(t)|\right\}, \\
M'=\sup_{t\in[0,\tau]}\left\{\sup_{\mbf{X},\mbf{\beta}_{jk}}|\mbf{\beta}_{jk}^{\trans}\mbf{X}'(t)|+\sup_{\mbf{Z}}|\mbf{Z}'(t)|\right\},
\end{gather*}
both of which are finite under Condition 2. Again, we suppress the arguments $\mbf{X}$ and $\mbf{Z}$ in any cumulative transition intensity matrix of the form $\mbf{A}(t; \mbf{X}, \mbf{Z}, \mbf{b}, \mbf{\beta}, \mbf{\Omega})$. For any $\mbf{\theta}_1$, $\mbf{\theta}_2\in\Theta$, $\mbf{\Omega}_1$, $\mbf{\Omega}_2\in\mathcal{L}_w$, and $(t_1,t_2]\subset [0,\tau]$, we denote $\mbf{A}(t;\mbf{b},\mbf{\beta}_i,\mbf{\Omega}_i)$ by $\mbf{A}_i(t;\mbf{b})$, and denote $\mbf{P}(t_1,t_2;\mbf{b},\mbf{\beta}_i,\mbf{\Omega}_i)$ by $\mbf{P}_i(t_1,t_2;\mbf{b})$, $i=1,2$. For any $m\times n$ matrix $\mbf{B}$ with elements $B_{ij}$, we define its norm as $\|\mbf{B}\|_1=\sum_{i=1}^m\sum_{j=1}^n|B_{ij}|$. Since $\mbf{A}(t;\mbf{b},\mbf{\beta},\mbf{\Omega})$ has zero row sums by definition, 
\begin{equation} \label{dif_A}
\begin{aligned}
& \|\mbf{A}_1(t;\mbf{b})-\mbf{A}_2(t;\mbf{b})\|_1 \\
\le{} & 2\sum_{(j,k)\in\mathcal{D}}\left|\mbf{A}_{1}(t;\mbf{b})^{(j,k)}-\mbf{A}_{2}(t;\mbf{b})^{(j,k)}\right| \\
={} & 
\begin{multlined}[t]
2\sum_{(j,k)\in\mathcal{D}}\left|\int_0^t\left[\exp\{\mbf{\beta}_{1jk}^{\trans}\mbf{X}(s)+\mbf{b}^{\trans}\mbf{Z}(s)\}-\exp\{\mbf{\beta}_{2jk}^{\trans}\mbf{X}(s)+\mbf{b}^{\trans}\mbf{Z}(s)\}\right]d\Lambda_{1jk}(s)\right. \\
\left.+\int_0^t \exp\{\mbf{\beta}_{2jk}^{\trans}\mbf{X}(s)+\mbf{b}^{\trans}\mbf{Z}(s)\}d(\Lambda_{1jk}-\Lambda_{2jk})(s)\right|
\end{multlined} \\
\le{} & 
\begin{multlined}[t]
2\sum_{(j,k)\in\mathcal{D}}\left\{e^{M\|\mbf{b}\|}\int_0^t\left|\exp\{\mbf{\beta}_{1jk}^{\trans}\mbf{X}(s)\}-\exp\{\mbf{\beta}_{2jk}^{\trans}\mbf{X}(s)\}\right|d\Lambda_{1jk}(s)\right. \\
+\exp\{\mbf{\beta}_{2jk}^{\trans}\mbf{X}(t)+\mbf{b}^{\trans}\mbf{Z}(t)\}\times|\Lambda_{1jk}(t)-\Lambda_{2jk}(t)| \\
\left.+\int_0^t|\Lambda_{1jk}(s)-\Lambda_{2jk}(s)|\times\left|\frac{d}{ds}\exp\{\mbf{\beta}_{2jk}^{\trans}\mbf{X}(s)+\mbf{b}^{\trans}\mbf{Z}(s)\}\right|ds\right\}
\end{multlined} \\
\le{} & 
\begin{multlined}[t]
2\sum_{(j,k)\in\mathcal{D}}\left\{e^{M\|\mbf{b}\|}\times C_0\|\mbf{\beta}_{1jk}-\mbf{\beta}_{2jk}\|+e^{M+M\|\mbf{b}\|}\times|\Lambda_{1jk}(t)-\Lambda_{2jk}(t)|\right. \\
\left.+e^{M+M\|\mbf{b}\|}(M'+M'\|\mbf{b}\|)\times\|\Lambda_{1jk}-\Lambda_{2jk}\|_{L_1}\right\} 
\end{multlined} \\
\le{} & e^{C_1+C_2\|\mbf{b}\|}\sum_{(j,k)\in\mathcal{D}}\left\{\|\mbf{\beta}_{1jk}-\mbf{\beta}_{2jk}\|+|\Lambda_{1jk}(t)-\Lambda_{2jk}(t)|+\|\Lambda_{1jk}-\Lambda_{2jk}\|_{L_1}\right\},
\end{aligned}
\end{equation}
where $C_0$, $C_1$, and $C_2$ are some positive constants, and $\|\cdot\|_{L_r}$ denotes the $L_r$-norm with respect to the Lebesgue measure in $[0,\tau]$ for any $r\ge 1$.

For any $(j,k)\in\mathcal{D}$, by applying Duhamel's equation in Theorem \Rom{2}.6.2 of \citet{andersen1993statistical} and the integration by parts formula, 
\begin{align*}
& \left|\mbf{P}_1(t_1,t_2;\mbf{b})^{(j,k)}-\mbf{P}_2(t_1,t_2;\mbf{b})^{(j,k)}\right| \\
={} & \left|\left\{\int_{t_1}^{t_2}\mbf{P}_1(t_1,t;\mbf{b})\left[d\mbf{A}_1(t;\mbf{b})-d\mbf{A}_2(t;\mbf{b})\right]\mbf{P}_2(t,t_2;\mbf{b})\right\}^{(j,k)}\right| \\
\le{} & 
\begin{multlined}[t]
\left|\Bigl\{\mbf{P}_1(t_1,t_2;\mbf{b})\left[\mbf{A}_1(t_2;\mbf{b})-\mbf{A}_2(t_2;\mbf{b})\right]-\left[\mbf{A}_1(t_1;\mbf{b})-\mbf{A}_2(t_1;\mbf{b})\right]\mbf{P}_2(t_1,t_2;\mbf{b})\Bigr\}^{(j,k)}\right| \\
+\left|\left\{\int_{t_1}^{t_2}\mbf{P}_1(t_1,dt;\mbf{b})\left[\mbf{A}_1(t;\mbf{b})-\mbf{A}_2(t;\mbf{b})\right]\mbf{P}_2(t,t_2;\mbf{b})\right\}^{(j,k)}\right| \\
+\left|\left\{\int_{t_1}^{t_2}\mbf{P}_1(t_1,t;\mbf{b})\left[\mbf{A}_1(t;\mbf{b})-\mbf{A}_2(t;\mbf{b})\right]\mbf{P}_2(dt,t_2;\mbf{b})\right\}^{(j,k)}\right|,
\end{multlined} 
\end{align*}
where the differentials of the transition probability matrices are calculated by
\begin{align*}
\mbf{P}_1(t_1,dt;\mbf{b})={} & \mbf{P}_1(t_1,t+dt;\mbf{b})-\mbf{P}_1(t_1,t;\mbf{b}) \\
={} & \mbf{P}_1(t_1,t;\mbf{b})\times[\mbf{P}_1(t,t+dt;\mbf{b})-\mbf{I}] \\
={} & \mbf{P}_1(t_1,t;\mbf{b})d\mbf{A}_1(t;\mbf{b}), \\
\mbf{P}_2(dt,t_2;\mbf{b})={} & \mbf{P}_2(t+dt,t_2;\mbf{b})-\mbf{P}_2(t,t_2;\mbf{b}) \\
={} & [\mbf{I}-\mbf{P}_2(t,t+dt;\mbf{b})]\times \mbf{P}_2(t+dt,t_2;\mbf{b}) \\
={} & -d\mbf{A}_2(t;\mbf{b})\mbf{P}_2(t+,t_2;\mbf{b}).
\end{align*}
Since any transition probability matrix has all elements no larger than 1,  
\begin{align*}
& \left|\mbf{P}_1(t_1,t_2;\mbf{b})^{(j,k)}-\mbf{P}_2(t_1,t_2;\mbf{b})^{(j,k)}\right| \\
\le{} & 
\begin{multlined}[t]
\|\mbf{A}_1(t_2;\mbf{b})-\mbf{A}_2(t_2;\mbf{b})\|_1+\|\mbf{A}_1(t_1;\mbf{b})-\mbf{A}_2(t_1;\mbf{b})\|_1 \\
+ \int_{t_1}^{t_2}\|d\mbf{A}_1(t;\mbf{b})\|_1\times\|\mbf{A}_1(t;\mbf{b})-\mbf{A}_2(t;\mbf{b})\|_1 \\
+ \int_{t_1}^{t_2}\|\mbf{A}_1(t;\mbf{b})-\mbf{A}_2(t;\mbf{b})\|_1\times\|d\mbf{A}_2(t;\mbf{b})\|_1.
\end{multlined}
\end{align*}
We use $Q_{ijk}$ to denote the probability measure generated by $\Lambda_{ijk}(t)$, for $i=1, 2$ and $(j,k)\in\mathcal{D}$, and use $\|\cdot\|_{L_r(Q)}$ to denote the $L_r(Q)$-norm $\|f\|_{L_r(Q)}=\left(\int|f|^rdQ\right)^{1/r}$, for any probability measure $Q$ and $r\ge1$. By plugging in the upper bound for $\|\mbf{A}_1(t;\mbf{b})-\mbf{A}_2(t;\mbf{b})\|_1$ in \eqref{dif_A} into the last two integrals in the above inequality, we obtain
\begin{align*}
& \int_{t_1}^{t_2}\|d\mbf{A}_1(t;\mbf{b})\|_1\times\|\mbf{A}_1(t;\mbf{b})-\mbf{A}_2(t;\mbf{b})\|_1 \\
={} & 2\sum_{(j,k)\in\mathcal{D}}\int_{t_1}^{t_2}\|\mbf{A}_1(t;\mbf{b})-\mbf{A}_2(t;\mbf{b})\|_1\times \exp\left\{\mbf{\beta}_{1jk}^{\trans}\mbf{X}(t)+\mbf{b}^{\trans}\mbf{Z}(t)\right\}d\Lambda_{1jk}(t) \\
\le{} & 2e^{M+M\|\mbf{b}\|}\sum_{(j,k)\in\mathcal{D}}\int_{t_1}^{t_2}\|\mbf{A}_1(t;\mbf{b})-\mbf{A}_2(t;\mbf{b})\|_1d\Lambda_{1jk}(t) \\
\le{} & e^{C_3+C_4\|\mbf{b}\|}\sum_{(j,k),\,(l,r)\in\mathcal{D}}\left\{\|\mbf{\beta}_{1lr}-\mbf{\beta}_{2lr}\|+\int_{t_1}^{t_2}|\Lambda_{1lr}(t)-\Lambda_{2lr}(t)|d\Lambda_{1jk}(t)+\|\Lambda_{1lr}-\Lambda_{2lr}\|_{L_1}\right\} \\
\le{} & e^{C_5+C_6\|\mbf{b}\|}\sum_{(j,k),\,(l,r)\in\mathcal{D}}\biggl\{\|\mbf{\beta}_{1lr}-\mbf{\beta}_{2lr}\|+\|\Lambda_{1lr}-\Lambda_{2lr}\|_{L_1(Q_{1jk})}+\|\Lambda_{1lr}-\Lambda_{2lr}\|_{L_1}\biggr\}.
\end{align*}
Likewise,
\begin{align*}
& \int_{t_1}^{t_2}\|\mbf{A}_1(t;\mbf{b})-\mbf{A}_2(t;\mbf{b})\|_1\times\|d\mbf{A}_2(t;\mbf{b})\|_1 \\
\le{} & e^{C_5+C_6\|\mbf{b}\|}\sum_{(j,k),\,(l,r)\in\mathcal{D}}\biggl\{\|\mbf{\beta}_{1lr}-\mbf{\beta}_{2lr}\|+\|\Lambda_{1lr}-\Lambda_{2lr}\|_{L_1(Q_{2jk})}+\|\Lambda_{1lr}-\Lambda_{2lr}\|_{L_1}\biggr\},
\end{align*}
where $C_3$, $C_4$, $C_5$, and $C_6$ are some positive constants. Therefore, there exist some positive constants $C_7$ and $C_8$, such that
\begin{align*}
& \left|\mbf{P}_1(t_1,t_2;\mbf{b})^{(j,k)}-\mbf{P}_2(t_1,t_2;\mbf{b})^{(j,k)}\right| \\
\le{} & 
\begin{multlined}[t]
e^{C_7+C_8\|\mbf{b}\|}\sum_{(l,r)\in\mathcal{D}}\biggl\{\|\mbf{\beta}_{1lr}-\mbf{\beta}_{2lr}\|+|\Lambda_{1lr}(t_1)-\Lambda_{2lr}(t_1)|+|\Lambda_{1lr}(t_2)-\Lambda_{2lr}(t_2)|\biggr. \\
\biggl.+\|\Lambda_{1lr}-\Lambda_{2lr}\|_{L_1}+\sum_{(l',r')\in\mathcal{D}}\Bigl\{\|\Lambda_{1lr}-\Lambda_{2lr}\|_{L_1(Q_{1l'r'})}+\|\Lambda_{1lr}-\Lambda_{2lr}\|_{L_1(Q_{2l'r'})}\Bigr\}\biggr\}.
\end{multlined}
\end{align*}
It then follows from the mean-value theorem that
\begin{align*}
& \left|\prod_{l=1}^N\mbf{P}_1(\tau_{l-1},\tau_l;\mbf{b})^{(S_{l-1},S_l)}-\prod_{l=1}^N\mbf{P}_2(\tau_{l-1},\tau_l;\mbf{b})^{(S_{l-1},S_l)}\right| \\
\le{} & \sum_{l=1}^N\left|\mbf{P}_1(\tau_{l-1},\tau_l;\mbf{b})^{(S_{l-1},S_l)}-\mbf{P}_2(\tau_{l-1},\tau_l;\mbf{b})^{(S_{l-1},S_l)}\right| \\
\le{} & e^{C_9+C_{10}\|\mbf{b}\|}\sum_{(j,k)\in\mathcal{D}}\left\{\|\mbf{\beta}_{1jk}-\mbf{\beta}_{2jk}\|+\|\Lambda_{1jk}-\Lambda_{2jk}\|_{L_1(\mu_1)}\right\}
\end{align*}
for some positive constants $C_9$ and $C_{10}$ and some finite positive measure $\mu_1$.
Thus, the difference between two likelihood functions is
\begin{align*}
& |L(\mbf{\theta}_1, \mbf{\Omega}_1)-L(\mbf{\theta}_2, \mbf{\Omega}_2)| \\
\le{} & 
\begin{multlined}[t]
\int_{\mbf{b}}\left|\prod_{l=1}^N\mbf{P}_1(\tau_{l-1},\tau_l;\mbf{b})^{(S_{l-1},S_l)}-\prod_{l=1}^N\mbf{P}_2(\tau_{l-1},\tau_l;\mbf{b})^{(S_{l-1},S_l)}\right|\phi(\mbf{b};\mbf{\Sigma}(\mbf{\gamma}_1))d\mbf{b} \\
+ \int_{\mbf{b}}\prod_{l=1}^N\mbf{P}_2(\tau_{l-1},\tau_l;\mbf{b})^{(S_{l-1},S_l)}\times \left|\phi(\mbf{b};\mbf{\Sigma}(\mbf{\gamma}_1))-\phi(\mbf{b};\mbf{\Sigma}(\mbf{\gamma}_2))\right|d\mbf{b}
\end{multlined} \\
\le{} & O(1)\biggl[\sum_{(j,k)\in\mathcal{D}}\Bigl\{\|\mbf{\beta}_{1jk}-\mbf{\beta}_{2jk}\|+\|\Lambda_{1jk}-\Lambda_{2jk}\|_{L_1(\mu_1)}\Bigr\}+\|\mbf{\gamma}_1-\mbf{\gamma}_2\|\biggr] \\
\le{} & O(1)\biggl\{\|\mbf{\theta}_1-\mbf{\theta}_2\|+\sum_{(j,k)\in\mathcal{D}}\|\Lambda_{1jk}-\Lambda_{2jk}\|_{L_1(\mu_1)}\biggr\},
\end{align*}
so
\[
\|L(\mbf{\theta}_1, \mbf{\Omega}_1)-L(\mbf{\theta}_2, \mbf{\Omega}_2)\|^2_{L_2(\mathbb{P})}\le O(1)\biggl\{\|\mbf{\theta}_1-\mbf{\theta}_2\|^2+\sum_{(j,k)\in\mathcal{D}}\|\Lambda_{1jk}-\Lambda_{2jk}\|^2_{L_2(\mu_2)}\biggr\},
\]
for some finite positive measure $\mu_2$.

Let $d$ be the dimension of $\Theta$, which is equal to $|\mathcal{D}|\times d_1+d_3$. There are at most $O(\epsilon^{-d})$ $\epsilon$-brackets covering $\Theta$, with respect to the Euclidean norm. Since each component of $\mathcal{L}_w$ consists of functions with bounded total variations in $[0,\tau]$, 
Lemma 2.2 of \citet{geer2000empirical} entails that the bracketing entropy of $\mathcal{L}_w$ satisfies $\log N_{[]}(\epsilon, \mathcal{L}_w, L_2(\mu_2))\lesssim\epsilon^{-1}$,
where $x\lesssim y$ means that $x\le cy$ for a positive constant $c$. Thus, there are a total of $\exp\{O(\epsilon^{-1})\}\times O(\epsilon^{-d})$ brackets $[\mbf{\theta}_1, \mbf{\theta}_2]\times [\mbf{\Omega}_1, \mbf{\Omega}_2]$ covering $\Theta\times\mathcal{L}_w$, such that 
\[
\|L(\mbf{\theta}_1, \mbf{\Omega}_1)-L(\mbf{\theta}_2, \mbf{\Omega}_2)\|^2_{L_2(\mathbb{P})}<O(\epsilon).
\]
Hence, the class of functions $\left\{L(\mbf{\theta}, \mbf{\Omega}):\,(\mbf{\theta}, \mbf{\Omega})\in\Theta\times\mathcal{L}_w\right\}$ is Donsker. Since $L(\mbf{\theta}_0, \mbf{\Omega}_0)$ is bounded away from zero, 
the preservation of the Donsker property under Lipschitz transformations implies that the class $\mathcal{M}$ is also Donsker.
\end{proof}

\begin{lemma} \label{lemma: rate}
Under Conditions 1-6, 
\[
E\left[\sum_{(j,k)\in\mathcal{D}}\sum_{l=1}^N\left\{\widehat{\Lambda}_{jk}(\tau_l)-\Lambda_{0jk}(\tau_l)\right\}^2\right]=O_p(n^{-2/3})+O\left(\|\widehat{\mbf{\beta}}-\mbf{\beta}_0\|^2+\|\widehat{\mbf{\gamma}}-\mbf{\gamma}_0\|^2\right).
\]
\end{lemma}
\begin{proof}
We have shown in Theorem 1 that $\widehat{\mbf{\Omega}}\in\mathcal{L}_w$ and $m(\widehat{\mbf{\theta}},\widehat{\mbf{\Omega}})\in\mathcal{M}$. Define
\[
J(\delta)=\int_0^\delta\sqrt{1+\log N_{[]}(\epsilon, \mathcal{M}, L_2(\mathbb{P}))}\,d\epsilon.
\]
It is easy to show that $J(\delta)\le O(\sqrt{\delta})$. By Lemma 1.3 of \citet{geer2000empirical} and the mean value theorem, 
\[
\mathbb{P}\left\{m(\mbf{\theta},\mbf{\Omega})-m(\mbf{\theta}_0,\mbf{\Omega}_0)\right\}\lesssim -H^2\left\{(\mbf{\theta},\mbf{\Omega}), (\mbf{\theta}_0,\mbf{\Omega}_0)\right\},
\]
where $H\left\{(\mbf{\theta},\mbf{\Omega}), (\mbf{\theta}_0,\mbf{\Omega}_0)\right\}$ is the Hellinger distance
\[
H\left\{(\mbf{\theta},\mbf{\Omega}), (\mbf{\theta}_0,\mbf{\Omega}_0)\right\}=\left\{\int\left[\sqrt{L(\mbf{\theta},\mbf{\Omega})}-\sqrt{L(\mbf{\theta}_0,\mbf{\Omega}_0)}\right]^2d\mu\right\}^{1/2}
\]
with respect to the dominating measure $\mu$.

The above results, together with the fact that $(\widehat{\mbf{\theta}},\widehat{\mbf{\Omega}})$ maximizes $\mathbb{P}_nm(\mbf{\theta},\mbf{\Omega})$ and the consistency result in Theorem 1, imply that all conditions in Theorem 3.4.1 of \citet{vaart1996weak} hold. Therefore, $H\left\{(\widehat{\mbf{\theta}},\widehat{\mbf{\Omega}}), (\mbf{\theta}_0,\mbf{\Omega}_0)\right\}=O_p(r_n^{-1})$, where $r_n$ satisfies $r_n^2J(r_n^{-1})\lesssim\sqrt{n}$. In particular, we choose $r_n$ in the order of $n^{1/3}$ such that $H\left\{(\widehat{\mbf{\theta}},\widehat{\mbf{\Omega}}), (\mbf{\theta}_0,\mbf{\Omega}_0)\right\}=O_p(n^{-1/3})$. By the mean value theorem, 
\begin{align*}
& \begin{multlined}[t]
E\left[\left(\int_{\mbf{b}}\left\{\prod_{l=1}^N\mbf{P}(\tau_{l-1}, \tau_l;\mbf{b},\widehat{\mbf{\beta}}, \widehat{\mbf{\Omega}})^{(S_{l-1}, \,S_l)}\right\}\phi(\mbf{b};\mbf{\Sigma}(\widehat{\mbf{\gamma}}))d\mbf{b}\right.\right. \\
\left.\left.-\int_{\mbf{b}}\left\{\prod_{l=1}^N\mbf{P}(\tau_{l-1}, \tau_l;\mbf{b},\mbf{\beta}_0, \mbf{\Omega}_0)^{(S_{l-1}, \,S_l)}\right\}\phi(\mbf{b};\mbf{\Sigma}(\mbf{\gamma}_0))d\mbf{b}\right)^2\right] 
\end{multlined} \\
={} & O_p(n^{-2/3}).
\end{align*}
Let $\mbf{\Sigma}_0=\mbf{\Sigma}(\mbf{\gamma}_0)$. By the mean value theorem and Proposition \Rom{2}.8.7 in \citet{andersen1993statistical},
\begin{align*}
& O_p(n^{-2/3})+O(1)\|\widehat{\mbf{\beta}}-\mbf{\beta}_0\|^2+O(1)\|\widehat{\mbf{\gamma}}-\mbf{\gamma}_0\|^2 \\
\ge{} & E\left[\left(\int_{\mbf{b}}\left\{\prod_{l=1}^N\mbf{P}(\tau_{l-1}, \tau_l;\mbf{b},\mbf{\beta}_0, \widehat{\mbf{\Omega}})^{(S_{l-1}, \,S_l)}-\prod_{l=1}^N\mbf{P}(\tau_{l-1}, \tau_l;\mbf{b},\mbf{\beta}_0, \mbf{\Omega}_0)^{(S_{l-1}, \,S_l)}\right\}\phi(\mbf{b};\mbf{\Sigma}_0)d\mbf{b}\right)^2\right]  \\
\ge{} & 
\begin{multlined}[t]
C_{1}E\left[\vphantom{\left(\sum_{l=1}\right)^2}\left(\sum_{l=1}^N\int_{\mbf{b}}\left\{\prod_{l'\ne l}\mbf{P}(\tau_{l'-1}, \tau_{l'};\mbf{b},\mbf{\beta}_0, \mbf{\Omega}_0)^{(S_{l'-1}, \,S_{l'})}\right\}\right.\right. \\
\left.\left.\times\left\{\int_{\tau_{l-1}}^{\tau_l}\mbf{P}(\tau_{l-1}, t;\mbf{b},\mbf{\beta}_0,\mbf{\Omega}_0)d\mbf{A}(t;\mbf{b},\mbf{\beta}_0,\widehat{\mbf{\Omega}}-\mbf{\Omega}_0)\mbf{P}(t,\tau_l;\mbf{b},\mbf{\beta}_0,\mbf{\Omega}_0)\right\}^{(S_{l-1},\,S_l)}\phi(\mbf{b};\mbf{\Sigma}_0)d\mbf{b}\right)^2\right]
\end{multlined}
\end{align*}
for some positive constant $C_{1}$. 
Define the metric space 
\[
\mathcal{V} =  \left\{\mbf{g}=\{g_{jk}\}_{(j,k)\in\mathcal{D}}:\; \|g_{jk}\|_{L_2(\mathbb{P})}<\infty,\, g_{jk}(0)=0\right\}
\]
and the norm 
\[
\|\mbf{g}\|_1=\left\{E\left[\sum_{(j,k)\in\mathcal{D}}\sum_{l=1}^Ng_{jk}(\tau_l)^2\right]\right\}^{1/2}.
\]
The space $(\mathcal{V}, \|\cdot\|_1)$ is a Banach space.
In addition, we define the seminorm 
\[
\|\mbf{g}\|_2=\left\{E\left[\left(\sum_{l=1}^N\sum_{(j,k)\in\mathcal{D}}\left\{R_{ljk}(\tau_l)g_{jk}(\tau_l)-\int_0^{\tau_l}g_{jk}(t)dR_{ljk}(t)\right\}\right)^2\right]\right\}^{1/2},
\]
where 
\[
R_{ljk}(t)=
\begin{multlined}[t]
\int_{\mbf{b}}\left\{\prod_{l'\ne l}\mbf{P}(\tau_{l'-1}, \tau_{l'};\mbf{b},\mbf{\beta}_0, \mbf{\Omega}_0)^{(S_{l'-1}, \,S_{l'})}\right\} \\
\times \mbf{P}(\tau_{l-1}, t;\mbf{b},\mbf{\beta}_0,\mbf{\Omega}_0)^{(S_{l-1},\,j)}\exp\left\{\mbf{\beta}_{0jk}^{\trans}\mbf{X}(t)+\mbf{b}^{\trans}\mbf{Z}(t)\right\} \\
\times \left[\mbf{P}(t,\tau_{l};\mbf{b},\mbf{\beta}_0,\mbf{\Omega}_0)^{(k,\,S_{l})}-\mbf{P}(t,\tau_{l};\mbf{b},\mbf{\beta}_0,\mbf{\Omega}_0)^{(j,\,S_{l})}\right] \\
\times I(\tau_{l-1}<t\le\tau_l)\phi(\mbf{b};\mbf{\Sigma}_0)d\mbf{b}
\end{multlined}
\]
for any $l\in\{1,\dots,N\}$, $(j,k)\in\mathcal{D}$, and $t\in[0,\tau]$.
If $\|\mbf{g}\|_2=0$ for some $\mbf{g}\in\mathcal{V}$, then with probability one, 
\[
\sum_{l=1}^N\sum_{(j,k)\in\mathcal{D}}\left\{R_{ljk}(\tau_l)g_{jk}(\tau_l)-\int_0^{\tau_l}g_{jk}(t)dR_{ljk}(t)\right\} = 0, 
\]
which implies that $g_{jk}$ has bounded total variation over $[0,\tau]$. Thus, the above equation can be written as 
\begin{align*}
0={} & \sum_{l=1}^N\sum_{(j,k)\in\mathcal{D}}\int_{\tau_{l-1}}^{\tau_l} R_{ljk}(t)dg_{jk}(t) \\
={} & 
\begin{multlined}[t]
\sum_{l=1}^N\int_{\mbf{b}}\left\{\prod_{l'\ne l}\mbf{P}(\tau_{l'-1}, \tau_{l'};\mbf{b},\mbf{\beta}_0, \mbf{\Omega}_0)^{(S_{l'-1}, \,S_{l'})}\right\} \\
\times\left\{\int_{\tau_{l-1}}^{\tau_l}\mbf{P}(\tau_{l-1}, t;\mbf{b},\mbf{\beta}_0,\mbf{\Omega}_0)d\mbf{A}(t;\mbf{b},\mbf{\beta}_0,\mbf{g})\mbf{P}(t,\tau_l;\mbf{b},\mbf{\beta}_0,\mbf{\Omega}_0)\right\}^{(S_{l-1},\,S_l)}\phi(\mbf{b};\mbf{\Sigma}_0)d\mbf{b} 
\end{multlined} \\
={} & \lim_{\epsilon\rightarrow0}\frac{1}{\epsilon}\int_{\mbf{b}}\left\{\prod_{l=1}^N\mbf{P}(\tau_{l-1}, \tau_l;\mbf{b},\mbf{\beta}_0, \mbf{\Omega}_0+\epsilon \mbf{g})^{(S_{l-1},\,S_l)}-\prod_{l=1}^N\mbf{P}(\tau_{l-1}, \tau_l;\mbf{b},\mbf{\beta}_0, \mbf{\Omega}_0)^{(S_{l-1},\,S_l)}\right\}\phi(\mbf{b};\mbf{\Sigma}_0)d\mbf{b}.
\end{align*}
We evaluate the above equation at all possible $(S_1,S_2,\dots,S_{N-1})$ given the start and end states $(S_0, S_N)$. Taking the sum of the resulting equations yields
\[
\lim_{\epsilon\rightarrow0}\frac{1}{\epsilon}\int_{\mbf{b}}\left\{\mbf{P}(0, \tau_N;\mbf{b},\mbf{\beta}_0, \mbf{\Omega}_0+\epsilon \mbf{g})^{(S_{0},S_N)}-\mbf{P}(0, \tau_N;\mbf{b},\mbf{\beta}_0, \mbf{\Omega}_0)^{(S_{0},S_N)}\right\}\phi(\mbf{b};\mbf{\Sigma}_0)d\mbf{b}=0.
\]
The above equation holds for any arbitrary $\tau_N$ and feasible $(S_0,S_N)$, which covers the whole set $\mathcal{D}$ under Condition 3. Thus, for any $t\in[0,\tau]$,
\begin{align*}
0={} & \lim_{\epsilon\rightarrow0}\frac{1}{\epsilon}\int_{\mbf{b}}\left[\mbf{P}(0, t;\mbf{b},\mbf{\beta}_0, \mbf{\Omega}_0+\epsilon \mbf{g})-\mbf{P}(0, t;\mbf{b},\mbf{\beta}_0, \mbf{\Omega}_0)\right]\phi(\mbf{b};\mbf{\Sigma}_0)d\mbf{b} \\
={} & \int_{\mbf{b}}\int_0^t\mbf{P}(0,s;\mbf{b},\mbf{\beta}_0,\mbf{\Omega}_0)d\mbf{A}(s;\mbf{b},\mbf{\beta}_0,\mbf{g})\mbf{P}(s,t;\mbf{b},\mbf{\beta}_0,\mbf{\Omega}_0)\phi(\mbf{b};\mbf{\Sigma}_0)d\mbf{b}
\end{align*}
By Condition 6, for any $t\in[0,\tau]$ and any $(j,k)\in\mathcal{D}$,
\[
\mbf{A}(t;\mbf{b},\mbf{\beta}_0,\mbf{g})^{(j,k)}=\int_0^t\exp\left\{\mbf{\beta}_{0jk}^{\trans}\mbf{X}(s)+\mbf{b}^{\trans}\mbf{Z}(s)\right\}dg_{jk}(s)=0.
\]
Therefore, $g_{jk}(t)=0$ for any $t\in[0,\tau]$ and $(j,k)\in\mathcal{D}$. Hence, $\|\cdot\|_2$ is a norm in $\mathcal{V}$. 

Next we show that the space $(\mathcal{V}, \|\cdot\|_2)$ is also a Banach space.
We only need to show the completeness of the space.
Consider any arbitrary Cauchy sequence $\{\mbf{g}_n = \{g_{njk}\}_{(j,k)\in\mathcal{D}}\}$ in the space $(\mathcal{V}, \|\cdot\|_2)$. We have 
\[
\|\mbf{g}_n-\mbf{g}_m\|_2^2 = E\left[\left(\sum_{l=1}^N\sum_{(j,k)\in\mathcal{D}}\left\{R_{ljk}(\tau_l)[g_{njk}(\tau_l)-g_{mjk}(\tau_l)]-\int_0^{\tau_l}[g_{njk}(t)-g_{mjk}(t)]dR_{ljk}(t)\right\}\right)^2\right]
\]
converges to zero almost surely. Thus,
\[
\sum_{l=1}^N\sum_{(j,k)\in\mathcal{D}}\left\{R_{ljk}(\tau_l)[g_{njk}(\tau_l)-g_{mjk}(\tau_l)]-\int_0^{\tau_l}[g_{njk}(t)-g_{mjk}(t)]dR_{ljk}(t)\right\}
\]
converges to zero almost surely. 
Suppose that $N=1$. Since $\tau_1$ can take any value within $[0,\tau]$, the above convergence implies that 
\[
\sum_{(j,k)\in\mathcal{D}}\left\{R_{1jk}(t)[g_{njk}(t)-g_{mjk}(t)]-\int_0^{t}[g_{njk}(s)-g_{mjk}(s)]dR_{1jk}(s)\right\}\rightarrow 0
\] 
almost surely for any $t\in [0,\tau]$.
For any given $(j,k)\in\mathcal{D}$, consider two subjects with the same observations except that the state pair $(S_0, j)$ is feasible for one of them and is unfeasible for the other.
Subtracting the above equations for these two subjects yields
\[
\sum_{k': (j,k')\in\mathcal{D}}\left\{R_{1jk'}(t)[g_{njk'}(t)-g_{mjk'}(t)]-\int_0^{t}[g_{njk'}(s)-g_{mjk'}(s)]dR_{1jk'}(s)\right\}\rightarrow 0 
\]   
almost surely. 
Similarly, we consider two subjects with the same observations except that the state pair $(k, S_1)$ is feasible for one of them and is unfeasible for the other. 
Subtracting the above equations for these two subjects yields
\[
R_{1jk}(t)[g_{njk}(t)-g_{mjk}(t)]-\int_0^{t}[g_{njk}(s)-g_{mjk}(s)]dR_{1jk}(s)\rightarrow 0 
\]
almost surely.
Let $G_{njk}(t) = R_{1jk}(t)g_{njk}(t)-\int_0^{t}g_{njk}(s)dR_{1jk}(s)$. 
Then $G_{njk}(t)$ converges to some function $G_{jk}(t)\in\mathcal{V}$.
In addition, we can easily verify that 
\[
G_{njk}(t) = \frac{R_{1jk}^2(t)}{R_{1jk}^{\prime}(t)}\times\left\{\frac{\int_0^{t}g_{njk}(s)dR_{1jk}(s)}{R_{1jk}(t)}\right\}^{\prime}.
\]
Since $R_{1jk}(t)$ and $R_{1jk}^{\prime}(t)$ are both uniformly bounded over $[0,\tau]$ under Condition 2, we conclude that $\mbf{g}_n$ converges in $\mathcal{V}$. Hence, the space $(\mathcal{V}, \|\cdot\|_2)$ is a Banach space.

By the Cauchy–Schwarz inequality, for any $\mbf{g}\in\mathcal{V}$,
\begin{align*}
\|\mbf{g}\|_2^2={} & E\left[\left(\sum_{l=1}^N\sum_{(j,k)\in\mathcal{D}}\left\{R_{ljk}(\tau_l)g_{jk}(\tau_l)-\int_0^{\tau_l}g_{jk}(t)dR_{ljk}(t)\right\}\right)^2\right] \\
\le{} & 2E\left[\left(\sum_{l=1}^N\sum_{(j,k)\in\mathcal{D}}R_{ljk}(\tau_l)g_{jk}(\tau_l)\right)^2\right]+2E\left[\left(\sum_{l=1}^N\sum_{(j,k)\in\mathcal{D}}\int_0^{\tau_l}g_{jk}(t)dR_{ljk}(t)\right)^2\right] \\
\le{} & C_2E\left[\sum_{l=1}^N\sum_{(j,k)\in\mathcal{D}}g_{jk}(\tau_l)^2\right] \\
={} & C_{2}\|\mbf{g}\|_1^2,
\end{align*}
where $C_{2}$ is a positive constant. 
Then it follows from the open mapping theorem in Banach spaces that the norms $\|\cdot\|_1$ and $\|\cdot\|_2$ are equivalent, so $\|\mbf{g}\|_2\ge C_{3} \|\mbf{g}\|_1$ for some constant $C_{3}$. Therefore, 
\[
O_p(n^{-2/3})+O(1)\|\widehat{\mbf{\beta}}-\mbf{\beta}_0\|^2+O(1)\|\widehat{\mbf{\gamma}}-\mbf{\gamma}_0\|^2\ge C_{1}C_{3} E\left[\sum_{(j,k)\in\mathcal{D}}\sum_{l=1}^N\left\{\widehat{\Lambda}_{jk}(\tau_l)-\Lambda_{0jk}(\tau_l)\right\}^2\right].
\]
The desired result thus follows.
\end{proof}

\begin{lemma} \label{hstar}
Under Conditions 1-6, the following are true:
\begin{itemize}
  \item[(\rom{1})] The solution ${\mbf{h}}^*$ to the equation system $\ell_{\mbf{\Omega}}^*\ell_{\mbf{\Omega}}({\mbf{h}}^*)=\ell_{\mbf{\Omega}}^*\mbf{\ell}_{\mbf{\theta}}$ exists;
  \item[(\rom{2})] $\mbf{\ell}_{\mbf{\theta}}(\widehat{\mbf{\theta}},\widehat{\mbf{\Omega}})-\ell_{\mbf{\Omega}}(\widehat{\mbf{\theta}},\widehat{\mbf{\Omega}})({\mbf{h}}^*)$ belongs to a Donsker class and converges in $L_2(\mathbb{P})$ norm to $\mbf{\ell}_{\mbf{\theta}}-\ell_{\mbf{\Omega}}({\mbf{h}}^*)$;
  \item[(\rom{3})] The matrix $E\left[\left\{\mbf{\ell}_{\mbf{\theta}}-\ell_{\mbf{\Omega}}({\mbf{h}}^*)\right\}^{\otimes2}\right]$ is invertible.
\end{itemize}
\end{lemma}
\begin{proof}
Throughout the proof, a function is evaluated at $(\mbf{\theta}_0,\mbf{\Omega}_0)$ if the arguments $(\mbf{\theta},\mbf{\Omega})$ are missing. We first verify (\rom{1}). We equip $\mathcal{H}$ with an inner product defined by
\[
\langle {\mbf{h}}^{(1)}, {\mbf{h}}^{(2)}\rangle=\sum_{(j,k)\in\mathcal{D}}\int_0^\tau h^{(1)}_{jk}(t)h^{(2)}_{jk}(t)d\Lambda_{0jk}(t).
\]
By the definition of the adjoint operator, for any ${\mbf{h}}^{(1)}, {\mbf{h}}^{(2)}\in\mathcal{H}$, 
\begin{align*}
\langle \ell_{\mbf{\Omega}}^*\ell_{\mbf{\Omega}}(\mbf{h}^{(1)}), \mbf{h}^{(2)} \rangle ={} & \mathbb{P}\left\{\ell_{\mbf{\Omega}}({\mbf{h}}^{(1)})\ell_{\mbf{\Omega}}({\mbf{h}}^{(2)})\right\} \\
={} & \sum_{(j,k)\in\mathcal{D}}\int_0^\tau \Gamma_{jk}({\mbf{h}}^{(1)})(t)h^{(2)}_{jk}(t)d\Lambda_{0jk}(t) \\
={} & \langle \mbf{\Gamma}(\mbf{h}^{(1)}), \mbf{h}^{(2)} \rangle ,
\end{align*}
where $\mbf{\Gamma}(\mbf{h}) = \{\Gamma_{jk}(\mbf{h})\}_{(j,k)\in\mathcal{D}}$, and 
\[
\Gamma_{jk}(\mbf{h})(t)=\sum_{(j',k')\in\mathcal{D}}\int_0^\tau E\left[\sum_{l,l'=1}^N B_{ljk}(t)B_{l'j'k'}(s)\right]h_{j'k'}(s)d\Lambda_{0j'k'}(s).
\]
Thus, $\ell_{\mbf{\Omega}}^*\ell_{\mbf{\Omega}}(\mbf{h}) = \mbf{\Gamma}(\mbf{h})$ for $\mbf{h}\in\mathcal{H}$.
Likewise, 
\[
\mathbb{P}\left\{\mbf{\ell}_{\mbf{\theta}}\ell_{\mbf{\Omega}}({\mbf{h}})\right\} = \sum_{(j,k)\in\mathcal{D}}\int_0^\tau E\left[\sum_{l=1}^N B_{ljk}(t)\mbf{\ell}_{\mbf{\theta}}\right]h_{jk}(t)d\Lambda_{0jk}(t)
\]
implies that $\ell_{\mbf{\Omega}}^*\mbf{\ell}_{\mbf{\theta}} = \mbf{\Psi}$, where $\mbf{\Psi}=\{\mbf{\psi}_{jk}\}_{(j,k)\in\mathcal{D}}$ with $\mbf{\psi}_{jk}(t) = E[\sum_{l=1}^N B_{ljk}(t)\mbf{\ell}_{\mbf{\theta}}]$.
Therefore, solving the equation system in (\rom{1}) is equivalent to solving the integral equation system $\mbf{\Gamma}(\mbf{h}^*) = \mbf{\Psi}$. That is, for $(j,k)\in\mathcal{D}$,
\begin{equation} \label{int_eq}
\sum_{(j',k')\in\mathcal{D}}\int_0^\tau E\left[\sum_{l,l'=1}^N B_{ljk}(t)B_{l'j'k'}(s)\right]\mbf{h}^*_{j'k'}(s)d\Lambda_{0j'k'}(s) = \mbf{\psi}_{jk}(t).
\end{equation}
Let $H_{jkj'k'}(t,s) = \sum_{l,l'=1}^N B_{ljk}(t)B_{l'j'k'}(s)$. 
We examine $E[H_{jkj'k'}(t,s)]$ by considering $t\ge s$ and $t<s$. If $t\ge s$,
\begin{align*}
& E[H_{jkj'k'}(t,s)] \\
={} & 
\begin{multlined}[t]
E\Biggl[\sum_{l>l'}\int_{\tau_{l'-1}=0}^s\int_{\tau_l'=s\lor (\tau_{l'-1}+\eta_2)}^{t}\int_{\tau_{l-1}=\tau_{l'}}^t\int_{\tau_{l}=t\lor (\tau_{l-1}+\eta_2)}^{\tau}B_{ljk}(t)B_{l'j'k'}(s) \Biggr. \\
\times f_{l}(\tau_{l-1},\tau_l)f_{l'}(\tau_{l'-1}, \tau_{l'})\,d\tau_l\,d\tau_{l-1}\,d\tau_{l'}\,d\tau_{l'-1} \\
\Biggl.+\sum_{l=1}^N \int_{\tau_l=t}^{\tau} \int_{\tau_{l-1}=0}^{s\land(\tau_{l}-\eta_2)}B_{ljk}(t)B_{lj'k'}(s)f_l(\tau_{l-1},\tau_l)\,d\tau_{l-1}\,d\tau_l\Biggr].
\end{multlined}
\end{align*}
If $t<s$,
\begin{align*}
& E[H_{jkj'k'}(t,s)] \\
={} & 
\begin{multlined}[t]
E\Biggl[\sum_{l<l'}\int_{\tau_{l-1}=0}^t\int_{\tau_l=t\lor (\tau_{l-1}+\eta_2)}^{s}\int_{\tau_{l'-1}=\tau_l}^s\int_{\tau_{l'}=s\lor (\tau_{l'-1}+\eta_2)}^{\tau}B_{ljk}(t)B_{l'j'k'}(s) \Biggr. \\
\times f_{l}(\tau_{l-1},\tau_l)f_{l'}(\tau_{l'-1}, \tau_{l'})\,d\tau_{l'}\,d\tau_{l'-1}\,d\tau_l\,d\tau_{l-1} \\
\Biggl.+\sum_{l=1}^N \int_{\tau_{l-1}=0}^t\int_{\tau_l=s\lor (\tau_{l-1}+\eta_2)}^{\tau}B_{ljk}(t)B_{lj'k'}(s)f_l(\tau_{l-1},\tau_l)\,d\tau_l\,d\tau_{l-1}\Biggr],
\end{multlined}
\end{align*}
where $x\land y = \min(x,y)$ and $x\lor y = \max(x,y)$, and where $\eta_2$, the minimum gap between any two successive examination times, and $f_l(\cdot, \cdot)$, the conditional density function of $(\tau_{l-1},\tau_l)$ given $(N, \mbf{X}, \mbf{Z})$, are both as defined in Condition 4.
We split the integral on the left side of \eqref{int_eq} into two parts accordingly. Then differentiation of \eqref{int_eq} twice with respect to $t$ yields the following integral equation system:
\begin{equation} \label{second_deriv}
\sum_{(j',k')\in\mathcal{D}}\left\{r_{1jkj'k'}(t)\mbf{h}_{j'k'}^*(t)+\int_0^t r_{2jkj'k'}(t,s)\mbf{h}_{j'k'}^*(s)ds+\int_t^{\tau} r_{3jkj'k'}(t,s)\mbf{h}_{j'k'}^*(s)ds\right\} = \mbf{\psi}_{jk}^{\prime\prime}(t),
\end{equation}
where 
\begin{align*}
r_{1jkj'k'}(t) ={} & \left[\frac{\partial}{\partial t}E[H_{jkj'k'}(t,s)]\Big\vert_{s=t-} 
-\frac{\partial}{\partial t}E[H_{jkj'k'}(t,s)]\Big\vert_{s=t+}\right]\Lambda^{\prime}_{0j'k'}(t) \\
r_{2jkj'k'}(t,s) ={} & \frac{\partial^2}{\partial t^2}E[H_{jkj'k'}(t,s)]\Big\vert_{t\ge s}\Lambda^{\prime}_{0j'k'}(s) \\
r_{3jkj'k'}(t,s) ={} & \frac{\partial^2}{\partial t^2}E[H_{jkj'k'}(t,s)]\Big\vert_{t< s}\Lambda^{\prime}_{0j'k'}(s). 
\end{align*}
By simple calculations,
\[
r_{1jkj'k'}(t) = -\left\{E\left[\sum_{l=1}^N B_{ljk}(t)B_{lj'k'}(t)\mid \tau_{l}=t\right]+E\left[\sum_{l=1}^N B_{ljk}(t)B_{lj'k'}(t)\mid \tau_{l-1}=t\right]\right\}\Lambda^{\prime}_{0j'k'}(t).
\]
For $l=1,\dots, N$, let $\boldsymbol{B}_{l}(t)=\left\{B_{l j k}(t)\right\}_{(j, k) \in \mathcal{D}}$. We also define $\mbf{G}(t)$ to be a $|\mathcal{D}|\times|\mathcal{D}|$ diagonal matrix with elements $\{\Lambda^{\prime}_{0jk}(t)\}_{(j,k)\in\mathcal{D}}$. 
Then the $|\mathcal{D}|\times|\mathcal{D}|$ matrix with the $\{(j,k),(j',k')\}$th element being $r_{1jkj'k'}$ can be expressed as 
\[
-\left\{E\left[\sum_{l=1}^{N} \boldsymbol{B}_{l}(t)^{\otimes 2} \mid \tau_{l-1}=t\right]+E\left[\sum_{l=1}^{N} \boldsymbol{B}_{l}(t)^{\otimes 2} \mid \tau_{l}=t\right]\right\}\times \mbf{G}(t),
\] 
which is negative definite. 
Thus, the linear operator of $\mbf{h}^*$ on the left side of \eqref{second_deriv} is the sum of an invertible operator and two compact operators. 
By Theorem 4.25 of \citet{rudin1973functional}, showing the invertibility of this operator is equivalent to showing that it is one to one. Since it is the second derivative of $\mbf{\Gamma}$, it suffices to show that $\mbf{\Gamma}$ is one to one.  
If $\mbf{\Gamma}(\mbf{h}) = 0$, then $\langle \mbf{\Gamma}(\mbf{h}), \mbf{h} \rangle = \mathbb{P}\{\ell_{\mbf{\Omega}}({\mbf{h}})\ell_{\mbf{\Omega}}({\mbf{h}})\} = 0$. Thus, with probability one,
\begin{align*}
0={} & \ell_{\mbf{\Omega}}({\mbf{h}}) \\
={} & 
\lim_{\epsilon\rightarrow0}\frac{1}{\epsilon}\int_{\mbf{b}}\Biggl\{\prod_{l=1}^N\mbf{P}(\tau_{l-1}, \tau_l;\mbf{b},\mbf{\beta}_0, \mbf{\Omega}_{0\epsilon, \mbf{h}})^{(S_{l-1},S_l)}-\prod_{l=1}^N\mbf{P}(\tau_{l-1}, \tau_l;\mbf{b},\mbf{\beta}_0, \mbf{\Omega}_0)^{(S_{l-1},S_l)}\Biggr\} 
\phi(\mbf{b};\mbf{\Sigma}_0)d\mbf{b},
\end{align*}
where $\mbf{\Sigma}_0 = \mbf{\Sigma}(\mbf{\gamma}_0)$. 
We evaluate the above equation at all possible $(S_1,S_2,\dots,S_{N-1})$ given the start and end states $(S_0, S_N)$. By taking the sum of the resulting equations, we obtain 
\[
\lim_{\epsilon\rightarrow0}\frac{1}{\epsilon}\int_{\mbf{b}}\left\{\mbf{P}(0, \tau_N;\mbf{b},\mbf{\beta}_0, \mbf{\Omega}_{0\epsilon, \mbf{h}})^{(S_{0},S_N)}-\mbf{P}(0, \tau_N;\mbf{b},\mbf{\beta}_0, \mbf{\Omega}_0)^{(S_{0},S_N)}\right\}\phi(\mbf{b};\mbf{\Sigma}_0)d\mbf{b} =0.
\]
The above equation holds for any $\tau_N\in[0,\tau]$ and any feasible $(S_0,S_N)$, which covers the whole set $\mathcal{D}$ under Condition 3. Thus, for any $t\in[0,\tau]$,
\begin{align*}
\mbf{0}={} & \lim_{\epsilon\rightarrow0}\frac{1}{\epsilon}\int_{\mbf{b}}\left[\mbf{P}(0, t;\mbf{b},\mbf{\beta}_0, \mbf{\Omega}_{0\epsilon, \mbf{h}})-\mbf{P}(0, t;\mbf{b},\mbf{\beta}_0, \mbf{\Omega}_0)\right]\phi(\mbf{b};\mbf{\Sigma}_0)d\mbf{b} \\
={} & \int_{\mbf{b}}\int_0^t\mbf{P}(0,s;\mbf{b},\mbf{\beta}_0,\mbf{\Omega}_0)d\mbf{A}(s;\mbf{b},\mbf{\beta}_0,\int \mbf{h}d\mbf{\Omega}_0)\mbf{P}(s,t;\mbf{b},\mbf{\beta}_0,\mbf{\Omega}_0)\phi(\mbf{b};\mbf{\Sigma}_0)d\mbf{b}.
\end{align*}
By Condition 6, for $(j,k)\in\mathcal{D}$,
\[
\int_0^t\exp\left\{\mbf{\beta}_{0jk}^{\trans}\mbf{X}(s)+\mbf{b}^{\trans}\mbf{Z}(s)\right\}h_{jk}(s)d\Lambda_{0jk}(s)=0.
\]
Differentiating the two sides with respect to $t$ yields $h_{jk}\equiv0$, which implies $\mbf{h}\equiv\mbf{0}$. 
Hence, the linear operator of \eqref{second_deriv} is invertible and thus the solution $\mbf{h}^*$ exists.

Next we verify (\rom{2}). 
Under Conditions 1, 2 and 4, $r_{ijkj'k'}$ ($i=1,2,3$, $(j,k), (j',k')\in\mathcal{D}$) and $\mbf{\psi}_{jk}^{\prime\prime}$ ($(j,k)\in\mathcal{D}$) in \eqref{second_deriv} are all continuously differentiable functions. 
Therefore, for $(j,k)\in\mathcal{D}$, ${\mbf{h}}^*_{jk}$ is continuously differentiable on $[0,\tau]$, which implies bounded variation. 
By the arguments in the proof of Lemma 1, 
$\mbf{\ell}_{\mbf{\theta}}(\widehat{\mbf{\theta}},\widehat{\mbf{\Omega}})-\ell_{\mbf{\Omega}}(\widehat{\mbf{\theta}},\widehat{\mbf{\Omega}})({\mbf{h}}^*)$ belongs to a Donsker class and converges in $L_2(\mathbb{P})$ norm to $\mbf{\ell}_{\mbf{\theta}}-\ell_{\mbf{\Omega}}({\mbf{h}}^*)$, as stated in (\rom{2}).


Finally, we verify (\rom{3}). If the matrix $E[\{\mbf{\ell}_{\mbf{\theta}}-\ell_{\mbf{\Omega}}({\mbf{h}}^*)\}^{\otimes2}]$ is singular, then there exist vectors $\mbf{v}=(\mbf{v}_1, \mbf{v}_2)$ with $\mbf{v}_1 = \{\mbf{v}_{1jk}\}_{(j,k)\in\mathcal{D}}\in\mathbb{R}^{|\mathcal{D}|\times d_1}$ and $\mbf{v}_2\in\mathbb{R}^{d_3}$, such that ${\mbf{v}}^{\trans}E[\{\mbf{\ell}_{\mbf{\theta}}-\ell_{\mbf{\Omega}}({\mbf{h}}^*)\}^{\otimes2}]\mbf{v}=0$. It follows that, with probability one, the score function along the submodel $\{\mbf{\theta}_0+\epsilon \mbf{v}, d\mbf{\Omega}_{0\epsilon}({\mbf{v}}^{\trans}{\mbf{h}}^*)\}$ is zero, where $d\mbf{\Omega}_{0\epsilon}({\mbf{v}}^{\trans}{\mbf{h}}^*) = \{(1-\epsilon {\mbf{v}}^{\trans}{\mbf{h}}^*_{jk})d\Lambda_{0jk}\}_{(j,k)\in\mathcal{D}}$. Therefore, 
\[
\begin{split}
\lim_{\epsilon\rightarrow0}\frac{1}{\epsilon}\int_{\mbf{b}}\left\{\prod_{l=1}^N\mbf{P}(\tau_{l-1}, \tau_l;\mbf{b},\mbf{\beta}_0+\epsilon \mbf{v}_1, \mbf{\Omega}_{0\epsilon}({\mbf{v}}^{\trans}{\mbf{h}}^*))^{(S_{l-1},S_l)}-\prod_{l=1}^N\mbf{P}(\tau_{l-1}, \tau_l;\mbf{b},\mbf{\beta}_0, \mbf{\Omega}_0)^{(S_{l-1},S_l)}\right\} \\
\times\phi(\mbf{b};\mbf{\Sigma}_0)d\mbf{b}+\int_{\mbf{b}} \left\{\prod_{l=1}^N\mbf{P}(\tau_{l-1}, \tau_l;\mbf{b},\mbf{\beta}_0, \mbf{\Omega}_0)^{(S_{l-1},S_l)}\right\}{\mbf{v}}_2^{\trans}\phi^\prime_{\mbf{\gamma}}(\mbf{b};\mbf{\Sigma}_0)d\mbf{b}=0
\end{split}
\]
with probability one, where $\mbf{\Sigma}_0 = \mbf{\Sigma}(\mbf{\gamma}_0)$. We evaluate the above equation at all possible $(S_1,S_2,\dots,S_{N-1})$ given the start and end states $(S_0, S_N)$. Then taking the sum of the resulting equations yields 
\[
\begin{split}
\lim_{\epsilon\rightarrow0}\frac{1}{\epsilon}\int_{\mbf{b}}\left\{\mbf{P}(0, \tau_N;\mbf{b},\mbf{\beta}_0+\epsilon \mbf{v}_1, \mbf{\Omega}_{0\epsilon}({\mbf{v}}^{\trans}{\mbf{h}}^*))^{(S_{0},S_N)}-\mbf{P}(0, \tau_N;\mbf{b},\mbf{\beta}_0, \mbf{\Omega}_0)^{(S_{0},S_N)}\right\}\phi(\mbf{b};\mbf{\Sigma}_0)d\mbf{b} \\
+\int_{\mbf{b}} \left\{\mbf{P}(0, \tau_N;\mbf{b},\mbf{\beta}_0, \mbf{\Omega}_0)^{(S_{0},S_N)}\right\}{\mbf{v}}_2^{\trans}\phi^\prime_{\mbf{\gamma}}(\mbf{b};\mbf{\Sigma}_0)d\mbf{b}=0.
\end{split}
\]
The above equation holds for any arbitrary $\tau_N$ and feasible $(S_0,S_N)$. Thus, for any $t\in[0,\tau]$,
\begin{align*}
\mbf{0}={} & 
\begin{multlined}[t]
\lim_{\epsilon\rightarrow0}\frac{1}{\epsilon}\int_{\mbf{b}}\left[\mbf{P}(0, t;\mbf{b},\mbf{\beta}_0+\epsilon \mbf{v}_1, \mbf{\Omega}_{0\epsilon}({\mbf{v}}^{\trans}{\mbf{h}}^*))-\mbf{P}(0, t;\mbf{b},\mbf{\beta}_0, \mbf{\Omega}_0)\right]\phi(\mbf{b};\mbf{\Sigma}_0)d\mbf{b} \\
+\int_{\mbf{b}} \mbf{P}(0, t;\mbf{b},\mbf{\beta}_0, \mbf{\Omega}_0){\mbf{v}}_2^{\trans}\phi^\prime_{\mbf{\gamma}}(\mbf{b};\mbf{\Sigma}_0)d\mbf{b}
\end{multlined} \\
={} & 
\begin{multlined}[t]
\int_{\mbf{b}}\biggl\{\int_0^t\mbf{P}(0,s;\mbf{b},\mbf{\beta}_0,\mbf{\Omega}_0)d\widetilde{\mbf{A}}(s;\mbf{b},\mbf{\beta}_0,\mbf{\Omega}_0)\mbf{P}(s,t;\mbf{b},\mbf{\beta}_0,\mbf{\Omega}_0)\biggr. \\
\biggl.+\mbf{P}(0,t;\mbf{b},\mbf{\beta}_0,\mbf{\Omega}_0)\frac{{\mbf{v}}_2^{\trans}\phi^\prime_{\mbf{\gamma}}(\mbf{b};\mbf{\Sigma}_0)}{\phi(\mbf{b};\mbf{\Sigma}_0)}\biggr\}\phi(\mbf{b};\mbf{\Sigma}_0)d\mbf{b},
\end{multlined}
\end{align*}
where the off-diagonal elements of the matrix $d\widetilde{\mbf{A}}(s;\mbf{b},\mbf{\beta}_0,\mbf{\Omega}_0)$ are 
\[
d\widetilde{\mbf{A}}(s;\mbf{b},\mbf{\beta}_0,\mbf{\Omega}_0)^{(j,k)} = \left[\mbf{v}_{1jk}^{\trans}\mbf{X}(s)-\mbf{v}^{\trans}\mbf{h}^*_{jk}(s)\right]\exp\left\{\mbf{\beta}_{0jk}^{\trans}\mbf{X}(s)+\mbf{b}^{\trans}\mbf{Z}(s)\right\}d\Lambda_{0jk}(s)
\]
if $(j,k)\in\mathcal{D}$ and 0 otherwise, and the diagonal elements are
\[
d\widetilde{\mbf{A}}(s;\mbf{b},\mbf{\beta}_0,\mbf{\Omega}_0)^{(j,j)} = -\sum_{k\ne j} d\widetilde{\mbf{A}}(s;\mbf{b},\mbf{\beta}_0,\mbf{\Omega}_0)^{(j,k)}.
\]
By Condition 6, $\mbf{v}_2=\mbf{0}$ and 
\[
\int_0^t\left[\mbf{v}_{1jk}^{\trans}\mbf{X}(s)-\mbf{v}^{\trans}\mbf{h}^*_{jk}(s)\right]\exp\left\{\mbf{\beta}_{0jk}^{\trans}\mbf{X}(s)+\mbf{b}^{\trans}\mbf{Z}(s)\right\}d\Lambda_{0jk}(s)=0
\]
for $(j,k)\in\mathcal{D}$.
Differentiating the two sides with respect to $t$ yields 
$\mbf{v}_{1jk}^{\trans}\mbf{X}(t)-\mbf{v}^{\trans}\mbf{h}^*_{jk}(t)=0$. By Condition 2, $\mbf{v}_{1jk}=\mbf{0}$. Hence, the matrix $E[\{\mbf{\ell}_{\mbf{\theta}}-\ell_{\mbf{\Omega}}(\mbf{h}^*)\}^{\otimes2}]$ is invertible.
\end{proof}

\section{Additional Simulation Studies}

We conducted a series of simulation studies with more complex disease processes.
We considered a four-state model with possible transitions including 1 to 2, 2 to 3, 2 to 4, and 3 to 4. 
For each subject, we generated two time-independent covariates, $X_1\sim\text{Ber}(0.5)$ and $X_2\sim \text{Unif}(0,1)$, and random effect $b\sim N(0,\sigma^2)$ with $\sigma^2=0.8$.
We set $(\beta_{121},\beta_{122}) = (0.5,-0.5)$, $\Lambda_{12}(t) = \log(1+0.5t)$,
$(\beta_{231},\beta_{232}) = (0.4, 0.2)$, $\Lambda_{23}(t) = 0.5t$,
$(\beta_{241},\beta_{242}) = (0.3,0.5)$, $\Lambda_{24}(t) = 0.4t$,
$(\beta_{341},\beta_{342}) = (-0.3, 0.7)$, and $\Lambda_{34}(t) = 0.6t$. 
The initial state of each subject belonged to 1, 2, or 3, with probabilities 0.25, 0.5, or 0.25, respectively. 
We generated six potential examination times for each subject, with the first being $\text{Unif}(0,1)$, and the gap between any two successive examination times being $0.05+\text{Unif}(0,1)$. We set the study end time to be 3, beyond which no examinations occurred. 
As is shown in Table~\ref{res_sim2} and Figure~\ref{fig_sim2}, the proposed methods continue to perform well.

\bibliography{Bibliography-MM-MC-Supp}

\pagebreak

\begin{table}
\caption{Estimation of the regression parameters in the simulation studies with four states. \label{res_sim2}} 
\begin{center}
\footnotesize
\begin{tabular}{lrrrr@{\hspace{1em}}rrrr@{\hspace{1em}}rrrr} 
\toprule
& \multicolumn{4}{c}{$n=400$} & \multicolumn{4}{c}{$n=800$} & \multicolumn{4}{c}{$n=1600$} \\
\cmidrule(lr){2-5} \cmidrule(lr){6-9} \cmidrule(lr){10-13}
& Bias & SE & SEE & CP & Bias & SE & SEE & CP & Bias & SE & SEE & CP \\ 
\midrule
$\beta_{121}=0.5$ & 0.028 & 0.367 & 0.352 & 94.6 & 0.011 & 0.247 & 0.241 & 94.7 & 0.003 & 0.171 & 0.167 & 94.6 \\ 
  $\beta_{122}=-0.5$ & $-0.049$ & 0.642 & 0.612 & 94.3 & $-0.018$ & 0.432 & 0.417 & 94.7 & $-0.010$ & 0.296 & 0.287 & 94.6 \\ 
  $\beta_{231}=0.4$ & 0.031 & 0.308 & 0.273 & 92.1 & 0.022 & 0.203 & 0.190 & 93.1 & 0.009 & 0.138 & 0.133 & 94.0 \\ 
  $\beta_{232}=0.2$ & 0.016 & 0.534 & 0.474 & 92.1 & 0.002 & 0.354 & 0.329 & 93.3 & $-0.008$ & 0.242 & 0.229 & 93.7 \\ 
  $\beta_{241}=0.3$ & 0.017 & 0.542 & 0.289 & 91.0 & 0.001 & 0.229 & 0.201 & 92.8 & $-0.002$ & 0.150 & 0.142 & 93.9 \\ 
  $\beta_{242}=0.5$ & 0.041 & 0.646 & 0.499 & 90.8 & 0.015 & 0.390 & 0.346 & 92.9 & 0.008 & 0.256 & 0.242 & 94.1 \\ 
  $\beta_{341}=-0.3$ & $-0.015$ & 0.283 & 0.246 & 91.7 & $-0.005$ & 0.187 & 0.172 & 93.5 & $-0.001$ & 0.126 & 0.121 & 94.1 \\ 
  $\beta_{342}=0.7$ & 0.045 & 0.498 & 0.431 & 91.5 & 0.020 & 0.325 & 0.300 & 93.0 & 0.002 & 0.220 & 0.209 & 93.9 \\ 
  $\sigma^2=0.8$ & 0.123 & 0.417 & 0.358 & 90.5 & 0.043 & 0.239 & 0.233 & 94.2 & $-0.004$ & 0.155 & 0.156 & 95.7 \\ 
  \bottomrule
\end{tabular}
\end{center}
Note: Bias and SE denote the median bias and empirical standard error, respectively.
SEE denotes the median of the standard error estimator, and CP denotes the 
empirical coverage percentage of the 95\% confidence interval. 
The log transformation is used to construct the confidence interval for $\sigma^2$.
Each entry is based on 10,000 replicates.
\end{table}

\begin{figure}
\begin{center}
\includegraphics[width=\textwidth]{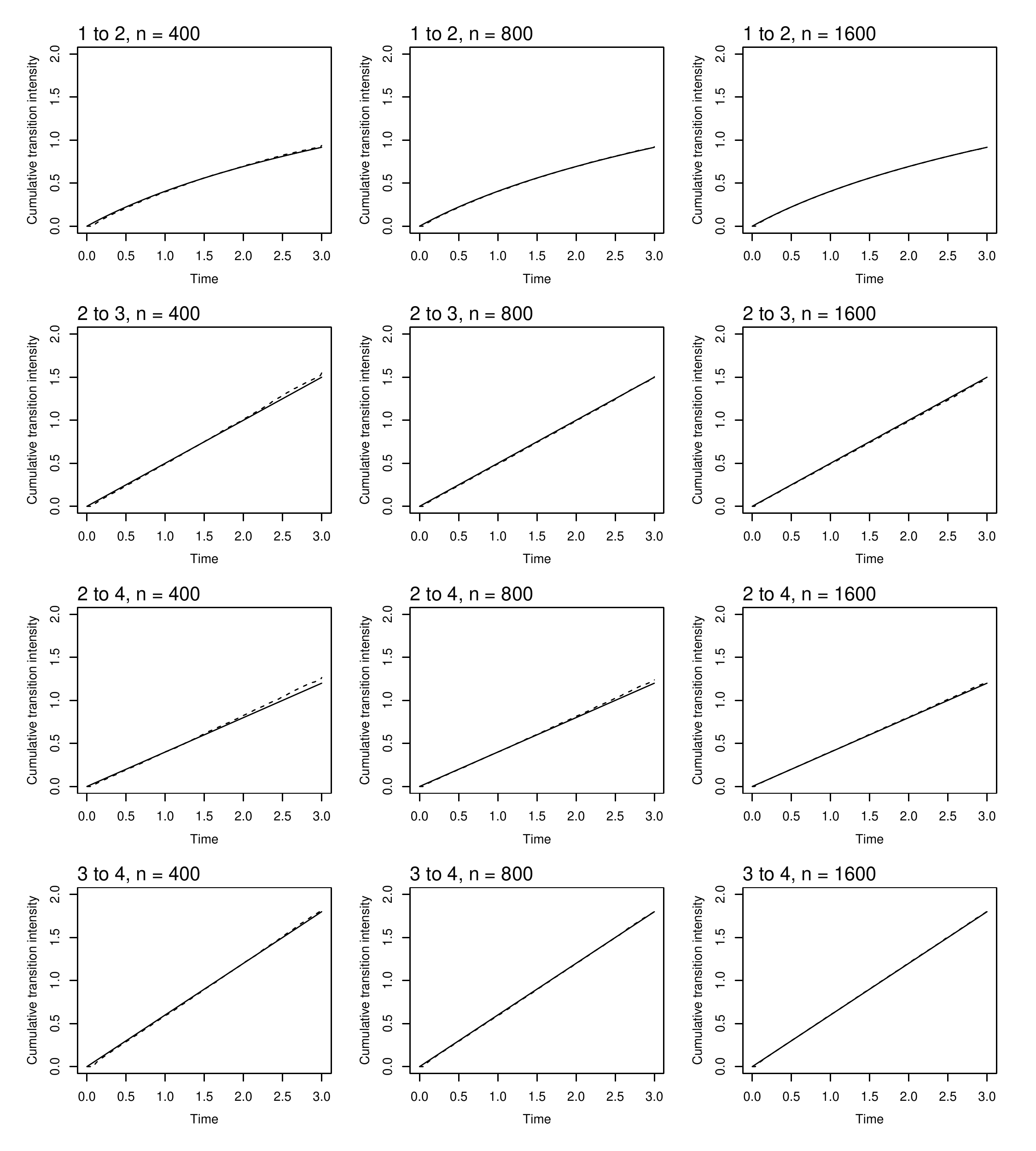}
\end{center}
\caption{Estimation of the cumulative transition intensities in the simulation studies with four states.
The solid and dashed curves show the true values and median estimates based on 10,000 replicates, respectively. \label{fig_sim2}} 
\end{figure}